\documentclass{llncs}

\usepackage{thmtools} 
\usepackage{thm-restate}
\usepackage{fancyvrb}
\usepackage{pdfpages}
 \usepackage{appendix}
\usepackage{pifont}
\usepackage[T1]{fontenc}
\usepackage{makecell}
\usepackage{tikz,pgfplots}
\usepackage{tikz-cd}
 \usepackage{listings}
 \usepackage{tabularx}
\usetikzlibrary{shadows,shapes,arrows,calc,automata,petri}
\usepackage{amssymb}

\usepackage{wrapfig}
\usepackage{url}
\usepackage{amsmath}
\usepackage{comment}
\usepackage{booktabs}
\usepackage{xspace}
\usepackage{multirow}
\usepackage{hyperref}
\usepackage[ruled,linesnumbered,noend]{algorithm2e}

\usepackage{pgf}
\pgfdeclarelayer{background}
\pgfsetlayers{background,main}
\usepackage{shuffle} % shuffle symbol %
\usepackage{enumitem} 
\usepackage{stmaryrd} % $\lightning$
\usepackage{amsfonts}
\usepackage{color}
\usetikzlibrary{matrix,decorations,calc,positioning,shapes,shadows,arrows,fit,automata}

\makeatletter
\tikzset{
  edge node/.code={%
      \expandafter\def\expandafter\tikz@tonodes\expandafter{\tikz@tonodes #1}}}
\makeatother
\tikzset{
  subseteq/.style={
    draw=none,
    edge node={node [sloped, allow upside down, auto=false]{$\subseteq$}}},
  Subseteq/.style={
    draw=none,
    every to/.append style={
      edge node={node [sloped, allow upside down, auto=false]{$\subseteq$}}}
  }
}
\tikzstyle{flostitle} =[draw,fill=white, text=black, rounded corners=2, inner sep=2pt]

\usepackage{xcolor}
 \usepackage{graphics}
\usepackage[font=small,labelfont=bf]{caption}
\usepackage{wrapfig}
\usepackage{soul}
\usepackage{fixme}
\sethlcolor{black}
\makeatletter
\newif\if@blind
\@blindtrue %use \@blindfalse on final version
\if@blind \sethlcolor{black}\else

\fi
\usepackage[pagewise]{lineno}
\usepackage{color, colortbl}
\newcounter{bulletscount}

\definecolor{mygreen}{rgb}{0,0.5,0}
\definecolor{forestgreen}{rgb}{0.13, 0.55, 0.13}

\usepackage{amsfonts}
\usepackage{amsopn}
\usepackage{complexity}
\usepackage{pifont}
\usepackage{microtype}

\newcommand{\peter}[1]{}
\newcommand{\flo}[1]{}

\newcommand\zthree  {\textsc{Z3}\@\xspace}
\newcommand\define{\mathrel{:=}}
\newcommand\redcross{{\color{red}{\ding{56}}}}

\newcommand{\setcond}[2]{\{ #1\mid \, #2\} }
\newcommand\gtick{{\color{mygreen}{\ding{52}}}}

\newcommand{\closedform}{$\phi$}
\newcommand{\closedm}{\phi}
\newcommand{\closedmt}{\phi_{t}}

\newcommand{\trans}{t^\Delta}
\newcommand{\mal}{\cdot}

\newcommand{\ineq}{$k\cdot m \geq c$ }
\newcommand{\ineqm}{k\cdot m \geq c }
\newcommand{\sumk}[1]{{k}_{#1}}
\newcommand{\mul}{\mathit{mul}}

\newcommand{\act}[1]{\mathit{ACT(#1)}}

\newcommand{\Real}{\mathbb{R}}
\newcommand{\ktdelta}{k\mal t^{\Delta}}

\newcommand{\Z}{\mathbb{Z}}
\newcommand{\N}{\mathbb{N}}
\newcommand{\fire}[1]{[ #1 \rangle}
\newcommand{\abs}[1]{\left\lvert #1 \right\rvert}

\newcommand{\tDelta}{t^\Delta}
\newcommand{\tMinus}{t^-}
\newcommand{\tPlus}{t^+}

\DeclareMathOperator{\ACT}{Act}
\DeclareMathOperator{\SOL}{Sol}

\newcommand{\Act}[1]{\ACT(#1)}
\newcommand{\Sol}[2]{\SOL(#1,#2)}

\newcommand{\cg}{\textsc{CGA}\@\xspace}
\newcommand{\ic}{\textsc{ICA}\@\xspace}

\newcommand{\inequalizer}{\textsc{Inequalizer}\@\xspace}
\newcommand{\mist}{\textsc{Mist}\@\xspace}
\newcommand{\amark}{m}
\newcommand{\avec}{v}
\newcommand{\avecp}{u}
\newcommand{\probnamereach}{\textsf{LSV(R)}}
\newcommand{\probnamecov}{\textsf{LSV(C)}}
\newcommand\bench[2][]{%
	\if\relax\detokenize{#1}\relax%
	\textsc{#2}\else \textsc{#2}\hspace{0.3pt}{\footnotesize(}#1{\footnotesize)}\fi\@\xspace}

\usetikzlibrary{plotmarks}
\pgfplotsset{linux/.style={%
		ylabel={\dartagnan},
		xlabel={\herd},
		xtick={0.01,0.1,1,10,100,1000},
		ytick={0.01,0.1,1,10,100,1000},
		xmax=1800,
		ymax=1800,
		xmin=0.01,
		ymin=0.01}}

\title{Petri Net Invariant Synthesis}
\author{Peter Chini\inst{1} \and Florian Furbach\inst{2}}
\institute{
	TU Braunschweig,
	\email{p.chini@tu-braunschweig.de}
	\and
	Uppsala University,
\email{florian.furbach@it.uu.se}
}

\begin{document}
\maketitle
%\flo{
%	input zu algorithmus 1 definieren?\\
%	ich glaub der begriff instrument the algorithm passt nicht wirklich, alternative?\\
%	offene intervalle? "(" statt +1 wäre näher an den formeln aber kann man kenntnisse über das offene interval vorraussetzen?\\
%	Warnings raus\\
%}
%\\
%
%Ideen zum kuerzen - momentan nicht nötig:
%\begin{itemize}
%	\item outline weg
%	\item incremental solving und minization weg in experimente
%	\item platz neben algorithmus 1 noch irgendwie nutzen
%\end{itemize}

\begin{abstract}
	We study the synthesis of inductive half spaces (IHS).
	These are linear inequalities that form inductive invariants for Petri nets, capable of disproving reachability or coverability.
	IHS generalize classic notions of invariants like traps or siphons.
	%This makes their synthesis a desirable tool for disproving reachability or coverability whenever traditional invariants may fail.
	Their synthesis is desirable for disproving reachability or coverability where traditional invariants may fail.
	
	We formulate a CEGAR-loop for the synthesis of IHS.
	The first step is to establish a structure theory of IHS.
	We analyze the space of IHS with methods from discrete mathematics and derive a linear constraint system closely over-approximating the space.
	To discard false positives, we provide an algorithm that decides whether a given half space is indeed inductive, a problem that we prove to be \coNP-complete.
	%We implemented the CEGAR-loop in the tool \inequalizer~and tested it against state-of-the-art techniques.
	%
	%For an efficient synthesis we first establish a structure theory of IHS.
	%The key is to analyze the space of all IHS with techniques from discrete mathematics.
	%We derive a linear constraint system, closely over-approximating the space.
	%We then derive a close over-approximation of this space, stated as a linear constraint system.
	%We formulate a CEGAR approach that utilizes this approximation. 
	%To discard false positives, we provide an algorithm that decides whether a given linear inequality is an IHS -- a problem that we prove to be \coNP-complete.
	%Finally, we formulate a CEGAR approach utilizing both the algorithm and the constraint system to synthesize IHS.
	We implemented the CEGAR-loop in the tool \inequalizer~and our experiments show that it is competitive against state-of-the-art techniques.
\end{abstract}
\section{Introduction}
A major task of today's program verification is to formulate and prove safety properties.
Such a property describes the desirable and undesirable behavior of a program, often expressed in terms of safe and unsafe states.
A safety property is satisfied if all executions of a program explore only safe states.
Phrased differently, it is violated if an unsafe state is reachable via an execution.	
%Deciding this is usually undecidable or at least rather complex already for simple properties~\cite{10.1007/3-540-58043-3_23}.
Testing reachability is usually a rather complex problem and often undecidable~\cite{Clarke2018,Hartmanis1967,Turing1937,Sipser1997}.

To restore decidability, the behavior of a program is often over-approximated.
Intuitively, an over-approximation describes a property that holds for all reachable states but fails for unsafe states.
Hence, over-approximations act like a separator between reachable and unsafe states and therefore provide a proof for the non-reachability of the latter.
%Under-approximations limit safety verification to a manageable part of executions.
%While they can show that an unsafe state is reachable, they cannot prove the absence of an execution leading to an unsafe state.
%To obtain such proofs, we study over-approximations that hold for the reachable states but fail for unsafe states.
%They separate safe and unsafe states and prove that no unsafe state is reachable.
Computing over-approximations is often achieved by generating some type of invariant~\cite{Floyd1967Flowcharts,Hoare:1969:ABC:363235.363259,Blanchet:2003:SAL:781131.781153,10.1007/978-3-540-30579-8_2}.
The challenge is to find a type that admits an efficient generation and that is expressive enough to separate reachable from unsafe states.
Inductive invariants are a prominent example~\cite{10.1007/978-3-540-69738-1_27,Cousot:1978:ADL:512760.512770,Gulwani:2008:PAC:1375581.1375616}.
If an inductive invariant holds for some state, then it also holds for any successor after a step of an execution.
Hence, if an inductive invariant is satisfied initially, it holds for all reachable states.

We generate inductive invariants for Petri nets, a well-established model of concurrent programs~\cite{Peterson:1981:PNT:539513,24143}.
Here, safety verification is usually expressed in terms of the Petri net reachability or coverability problem.
The former is known to be \ComplexityFont{ACKERMANN}-complete~\cite{Czerwinski2021,Leroux2021,Czerwinski2019,Leroux2019}, the latter is \EXPSPACE-complete~\cite{RACKOFF1978223,yale1976reachability,Cardoza:1976:ESC:800113.803630}.
Despite the ongoing algorithmic development, in particular for coverability~\cite{KARP1969147,10.1007/978-3-642-32940-1_35,10.1007/978-3-642-31131-4_12,10.1007/978-3-540-75596-8_9,10.1007/978-3-642-21834-7_5}, computational requirements of solving both problems often exceeds practical limits.
This has led to the development of classic Petri net invariants like \emph{traps}, \emph{siphons}, or \emph{place invariants}~\cite{Peterson:1981:PNT:539513} that may help to solve both problems more efficiently.
Typically, these invariants are based on linear dependencies of places or transitions and can be synthesized easily by incorporating tools and solvers from linear programming.
%Safety verification is usually expressed in terms of the Petri net reachability or coverability problem.
%Both problems are complex but decidable~\cite{grr,RACKOFF1978223,Cardoza:1976:ESC:800113.803630}.
%Reachability is known to be non-elementary-hard~\cite{Czerwinski2019,yale1976reachability,DBLP:conf/lics/BlondinFGHM15} and coverability is $\EXPSPACE$-complete~\cite{RACKOFF1978223,yale1976reachability,Cardoza:1976:ESC:800113.803630}.
%Despite the ongoing algorithmic development, especially for coverability~\cite{KARP1969147,10.1007/978-3-642-32940-1_35,10.1007/978-3-642-31131-4_12,10.1007/978-3-540-75596-8_9,10.1007/978-3-642-21834-7_5}, computational requirements of solving both problems often exceeds practical limits.
%This led to the development of invariants approximating the problems.
%Well known examples are \emph{place invariants}, \emph{traps}, and \emph{siphons}~\cite{Peterson:1981:PNT:539513}.
%The synthesis of these invariants is comparatively simple but we would like to generate a more expressive type.

The trade-off for the efficient synthesis of these classic invariants is that their expressiveness is limited and often not sufficient to prove non-reachability of a marking.
We study \emph{inductive half spaces} (IHS)~\cite{Sankaranarayanan2003,Triebel2015}, a type of invariants with increased expressiveness.
These consist of a tuple $(k,c)$, where $k$ is a vector over the places of the Petri net and $c$ is an integer.
The corresponding \emph{half space} is a subset of the space of markings, containing all markings $\amark$ that satisfy the inequality $\ineqm$.
It is called \emph{inductive} if the markings that are in the half space do not leave it after firing a transition.
Inductive half spaces generalize many of the classical Petri net invariants~\cite{Triebel2015} and preserve their \emph{linear nature}.
However, the synthesis of IHS remained an open problem.
%We study \emph{inductive half spaces} (IHS), a type of invariants that generalizes the mentioned notions~\cite{Sankaranarayanan2003,Triebel2015}.
%Recall that a place invariant is a vector $k$ over the places so that there is a $c \in \Z$ with $k \cdot m = c$ for each reachable marking $m$ of the Petri net.
%IHS relax the definition to an inequality by requiring $k \cdot m \geq c$.
%This has two implications.
%First, the constant $c$ is no longer uniquely determined by $k \cdot m$.
%It is now an independent part of the IHS.
%Second, $k \cdot m \geq c$ defines a half space within the set of markings.
%It contains all $m$ that satisfy the inequality.
%We refer to $(k,c)$ as \emph{half space}.
%Such a tuple is \emph{inductive}, and thus an IHS, if the markings that are in the half space do not leave it after firing a transition of the net.

Our contribution is a method for the synthesis of inductive half spaces.
More precise, we compute IHS that separate an initial marking $m_0$ from a final marking $m_f$, proving the latter non-reachable.
This task is formalized in the \emph{linear safety verification problem} $\probnamereach$.
Given $\amark_0$ and $\amark_f$, it asks for an IHS $(k,c)$ such that $k \cdot \amark_0 \geq c$ and $k \cdot \amark_f < c$.
The problem was first considered in~\cite{Sankaranarayanan2003} for continuous Petri nets.
The synthesis of IHS is much easier in the continuous case.
In fact, an entire subclass we call \emph{non-trivial inductive half spaces} does not occur in this setting.
So far, $\probnamereach$ has not been considered in its full generality and its decidability is still unknown.

We provide a semi-decision procedure for $\probnamereach$ using \emph{counter example guided abstraction refinement} (CEGAR)~\cite{CEGAR00}, a state-of-the-art technique in program verification.
We illustrate the approach
in \autoref{Figure:cegar}.
Suppose we are given
\begin{wrapfigure}[11]{r}{5.3cm}
	\vspace{-1cm}
	\begin{center}
\begin{tikzpicture}[>=stealth', bend angle=45, auto, ->, node distance=1.5cm]
\node[] (pn)  {Petri net $N$, markings $\amark_0, \amark_f$};
\node[below of = pn, draw,rounded corners] (smt)  {SMT-solver};
\node[below of = smt, draw,rounded corners] (algo)  {Checker};
%\node[right of = algo] (solt)  {};
\node[right of = algo, node distance = 2cm] (sol)  {\gtick \ineq};
%\node[left of = notkt] (notk)  {$\neg mul(k)$};
\node[right of = smt] (false)  {\redcross};

\draw[->, color=green] (algo) edge node {c} (sol);

\draw[->] (pn) edge[left] node {\closedform} (smt);
\draw[->, color=green] (smt) edge[out = 315, in = 45,right] node {$k$} (algo);
%\draw[->] (algo) edge[left] node {} (notk);
\draw[->, color=red] (algo) edge[out = 135, in = 225] node {$\neg mul(k)$} (smt);
\draw[->, color=red] (smt) edge[left] node {} (false);
%\draw[->] (notk) edge[left] node {} (smt);
\end{tikzpicture}
\end{center}
	\vspace{-0.5cm}
	\caption{The CEGAR loop.}
	\label{Figure:cegar}
\end{wrapfigure}
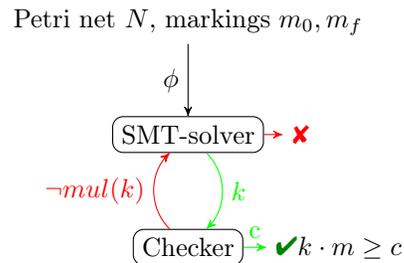
a Petri net $N$, an initial marking $\amark_0$, and a marking $\amark_f$ for which we want to disprove reachability from $\amark_0$.
Our approach attempts to synthesize an IHS that separates $m_f$ from the reachable markings of $N$.
It begins by constructing a formula $\phi$ of linear constraints from the given information and passes it to an SMT-solver.
Roughly, $\phi$ describes necessary conditions for solutions of $\probnamereach$.
For each solution $(k,c)$ of $\probnamereach$, the vector $k$ is a solution of $\phi$.
%Vice versa, if $k$ is a solution of $\phi$, there might exist a $c \in \Z$ such that $(k,c)$ solves $\probnamereach$.
If the SMT-solver does not find a solution to $\phi$, then no separating IHS exists.
Otherwise we find a vector $k$ of a half space candidate.
We then determine whether there exists a $c \in \Z$ such that $(k,c)$ is indeed inductive.
In order to synthesize such a $c$, we developed a \emph{constant generation algorithm} (\cg).
If \cg~is successful, we have found a separating IHS $(k,c)$.
Otherwise, $k$ does not admit a suitable constant $c$ and we apply a refinement.
We set $\phi = \phi \wedge \neg\mul(k)$, where $\neg\mul(k)$ is a linear constraint that excludes all multiples of $k$ and repeat the above process.
Note that the loop may not terminate.
But if it does, we obtain an answer to $\probnamereach$.

To realize the CEGAR-loop, we make the following main contributions.
\begin{itemize}
	\item We develop a structure theory of inductive half spaces.
	It decomposes the space of IHS into \emph{trivial} and \emph{non-trivial} half spaces.
	While the synthesis of trivial IHS is simple, the synthesis of non-trivial ones is challenging.
	By employing techniques from discrete mathematics, we can determine necessary conditions for non-trivial IHS and construct the required formula $\phi$.
	\item We present two algorithms: the \emph{inductivity checker} (\ic)  and the \emph{constant generation algorithm} (\cg).
	The former determines whether a given half space is inductive, a problem that we prove \coNP-complete.
	This answers an open question from~\cite{Triebel2015}.
	\ic~combines structural properties of IHS with dynamic programming.
	The algorithm \cg~synthesizes a constant $c$ for a solution $k$ of $\phi$.
	\cg~is an instrumentation of \ic.
	As termination argument, it uses an interesting connection between IHS and the \emph{Frobenius number}.
	\item We implemented the CEGAR-loop in the tool \inequalizer.
	Employing it, we disproved reachability and coverability for a benchmark of widely used concurrent programs. 
	The results are compared to algorithms implemented in \mist~\cite{mist} and show \inequalizer to be competitive.
\end{itemize}
\paragraph{Related Work}\label{sec:pnrelated}
The reachability problem of Petri nets is a central problem in theoretical computer science.
Its complexity was finally resolved after 45 years and proven to be \ComplexityFont{ACKERMANN}-complete.
The upper bound is due to Leroux and Schmitz~\cite{Leroux2019}.
The authors refined several classical algorithms for reachability like the one by Kosaraju~\cite{Kosaraju1982}, Mayr~\cite{Mayr1981,grr}, and Lambert~\cite{Lambert1992}.
Hardness was first considered by Lipton~\cite{yale1976reachability}. 
He proved reachability \EXPSPACE-hard.
Czerwinski et al~\cite{Czerwinski2019} improved the lower bound to non-elementary.
A new result due to Leroux, Czerwinski, and Orlikowski~\cite{Leroux2021,Czerwinski2021} closes the gap completely.

Many safety verification tasks can be phrased in terms of coverability.
The \EXPSPACE-completeness of the problem was determined by Rackoff \cite{RACKOFF1978223} and Lipton~\cite{yale1976reachability}.
Despite this, efficient algorithms keep getting developed~\cite{KARP1969147}.
Modern approaches are based on forward or backward state space exploration~\cite{zbMATH05506898,GEERAERTS2006180,10.1007/978-3-642-32940-1_35,10.1007/978-3-642-39799-8_10,10.1007/978-3-642-31131-4_12}.
A method that has drawn interest are so-called unfoldings~\cite{Esparza99anunfolding,McMillan:1995:TSS:204526.204528,Esparza:1996:IMU:646480.693793,Langerak:1999:CFP:647768.733925}. 
Notably, Abdulla et al~\cite{Abdulla2004} solve coverability by constructing an unfolding that represents backwards reachable states.
They analyze it using an SMT-formula.

Profiting from advances in SMT-solving, deriving program properties by constraint solving has become popular~\cite{gupta2009tests,Abdulla2004,10.1007/978-3-319-08867-9_40,Gulwani:2008:PAC:1375581.1375616}. 
In~\cite{10.1007/978-3-540-24622-0_20}, ranking functions are synthesized by solving linear inequalities.
%in order to prove termination of unnested programs.
Synthesis methods for Petri nets involving SMT-solving often simplify the task by using continuous values.
In~\cite{10.1007/978-3-319-08867-9_40}, Esparza et al generate inductive invariants disproving co-linear properties.
Sankaranarayanan et al~\cite{Sankaranarayanan2003} synthesize IHS over continuous 
Petri nets.
Compared to the latter, we generate a larger class of invariants: \emph{non-trivial} IHS do not occur in~\cite{Sankaranarayanan2003} but can be necessary for discrete nets (see~\autoref{Fig:Runexample}).
%For instance, an IHS for the net in~\autoref{Fig:Runexample} needs to be non-trivial.
The structure of IHS was first considered by Triebel and Sürmeli~\cite{Triebel2015}.
The authors show that IHS generalize notions like traps, siphons, and place invariants.
\paragraph{Outline}
In \autoref{Section:Prelim}, we introduce the necessary notions around Petri nets.
The structure of IHS is examined in \autoref{Section:Charinduct}.
In \autoref{Section:Generating}, we formulate the SMT-formula in the CEGAR loop.
The algorithms \ic~and \cg~are given in \autoref{Sec:checking}.
Experimental results are presented in \autoref{sec:pn:experiments}.
For brevity, we omit a number of formal proofs.
They can be found in the appendix.\flo{ref arxive}

\section{Linear Safety Verification}
\label{Section:Prelim}
We introduce the linear safety verification problems for Petri nets.
They formalize the question of whether there exists an inductive half space which disproves reachability or coverability of a certain marking.
To this end, we formally introduce half spaces and the necessary notions around Petri nets.
%---------------------------------------------------------------------
%---------------------------------------------------------------------
%---------------------------------------------------------------------
%---------------------------------------------------------------------
\paragraph{Petri Nets}
A \emph{Petri net} is a tuple $N = (P,T,F)$, where $P$ is a finite set of \emph{places}, $T$ is a finite set of \emph{transitions}, and $F: (P \times T) \cup (T \times P) \rightarrow \N$ is a \emph{flow function}.
We denote the number of places $\abs{P}$ by $n$. 
The places are numbered. For convenience, we use a place $p_i$ and their numeric value $i$ interchangeably: 
Given a vector $x\in \N^n$, we denote its $i$-th component as both $x(i)$ and $x(p_i)$. 
For a transition $t \in T$, we define vectors $\tMinus, \tPlus \in \N^n$. 
The $i$-th component of $\tMinus$, with $p_i \in P$, is defined to be $F(p_i,t)$, written $\tMinus(i)=\tMinus(p_i) \define F(p_i,t)$.
Similarly, $\tPlus(i) = \tPlus(p_i) \define F(t,p_i)$.
The vector $\tDelta$ captures the difference $\tDelta \define \tPlus - \tMinus$.

The semantics of a Petri net $N$ is defined in terms of markings.
A \emph{marking} $\amark$ is a vector in $\N^n$.
Intuitively, it puts a number of \emph{tokens} in each place.
A marking is said to \emph{enable} a transition $t$ if $\amark(p) \geq \tMinus(p)$ for each place $p \in P$, written $\amark \geq \tMinus$.
The set of all markings that enable $t$ is called the \emph{activation space} of $t$ and is denoted by $\ACT(t)$.
Note that $\ACT(t) = \setcond{\tMinus + \avec}{ \avec \geq 0 }$.
If $\amark \in \ACT(t)$, then $t$ can be \emph{fired}, resulting in the new marking $\amark' = \amark + \tDelta$.
This constitutes the \emph{firing relation}, written as $\amark \fire{t} \amark'$. 
We lift the relation to sequences of transitions $\sigma = t_1 \dots t_k \in T^*$ where convenient, writing $\amark \fire{\sigma} \amark'$. 
A marking $\amark_f$ is called \emph{reachable} from a marking $\amark_0$ if there is a sequence of transitions $\sigma$ such that $\amark_0 \fire{\sigma} \amark_f$. 
We use $\mathit{post}^*(\amark_0)$ to denote the markings reachable from $\amark_0$ and $\mathit{pre}^*(\amark_f)$ are the markings from which $\amark_f$ is reachable.
The \emph{upward closure} of $\amark_f$ is ${\uparrow\!\amark_f}=\setcond{\amark\in \N^n}{\amark\geq \amark_f}$.
A marking $\amark_f$ is \emph{coverable} from $\amark_0$, if there is a sequence of transitions $\sigma$ and an $\amark\in\ \uparrow\!\amark_f$ such that $\amark_0 \fire{\sigma} \amark$. 
%---------------------------------------------------------------------
%---------------------------------------------------------------------
%---------------------------------------------------------------------
%---------------------------------------------------------------------
\paragraph{(Inductive) Half Spaces}
We describe sets of markings by means of half spaces.
Let $N = (P,T,F)$ be a Petri net, $k \in \Z^n$ a vector, and $c \in \Z$ an integer. 
The \emph{half space} defined by $k$ and $c$ is $\Sol{k}{c} = \setcond{\amark \in \Z^n}{ k \cdot \amark \geq c}$.
Here, $k \cdot \amark = \sum_{p \in P} k(p) \cdot \amark(p)$ is the usual scalar product.
We also refer to the tuple $(k,c)$ as half space.
Note that we could also define half spaces via $k \cdot \amark \leq c$. 
This is of course equivalent since $k \cdot \amark \geq c$ if and only if $-k \cdot M \leq -c$. 
We are interested in half spaces that are inductive in the sense that they cannot be left by firing transitions.
A half space $(k,c)$ is \emph{$t$-inductive} if for any $\amark \in \Act{t} \cap \Sol{k}{c}$ we have $\amark + \tDelta \in \Sol{k}{c}$. 
A half space $(k,c)$ is \emph{inductive} if it is $t$-inductive for all $t \in T$. 
We use IHS as a shorthand for inductive half space.

A half space $(k,c)$ is not $t$-inductive if and only if it contains a marking $\amark$ with $k \cdot \amark  \geq c$ that enables $t$, i.e. $\amark  \geq \tMinus$, and from which we leave the half space by firing $t$: $k \cdot (\amark + \tDelta) < c$.
Since $\amark \geq \tMinus$ if and only if there is an $x\in \N^n$ with $\amark = \tMinus + x$, we can state inductivity in terms of an infeasibility requirement:
%---------------------------------------------------------------------
%---------------------------------------------------------------------
%---------------------------------------------------------------------
%---------------------------------------------------------------------
\begin{restatable}{theorem}{inductive}
\label{Theorem:inductive}
	A half space $(k,c)$ is $t$-inductive iff there is no vector $x \in \N^n$ with
	\begin{align*}
		c \leq k\cdot x + k \cdot t^- < c- k\cdot t^\Delta.
	\end{align*}
\end{restatable}
%---------------------------------------------------------------------
%---------------------------------------------------------------------
%---------------------------------------------------------------------
%---------------------------------------------------------------------
\autoref{Theorem:inductive} provides a way of disproving inductivity of a half space by finding a suitable vector $x$.
It is a key ingredient of our further development.
%It is the key ingredient of our algorithm in \autoref{Sec:checking}.
%This theorem provides us a way to disprove inductivity by finding a vector that satisfies two linear inequalities. It is a key insight behind our method for checking inductivity in \autoref{Sec:checking}.
%---------------------------------------------------------------------
%---------------------------------------------------------------------
%---------------------------------------------------------------------
%---------------------------------------------------------------------
\begin{figure}
		\centering
		\begin{minipage}{.47\textwidth}
			\centering
			\begin{tikzpicture}[>=stealth', bend angle=45, auto, ->, shorten >= 1pt]
	\node [place] (p1) {$p_1$};
	\coordinate[right of = p1, node distance = 2cm] (tco);
	\node[place, right of = tco, node distance = 2cm] (p2) {$p_2$};
	
	\node [transition] (t) at (tco) {$t$};
	\path (p1) edge[out = 315, in = 225] node[above] {2} (t);
	\path (t) edge[out = 135, in = 45] (p1);
	\path (p2) edge[out = 135, in = 45] (t);
	\path (t) edge[out = 315, in = 225] node[above] {2} (p2);

	\node [transition, above of = tco, node distance = 1.3cm] (u) {$u$};
	\path (p1) edge[out = 90, in = 180] (u);
	\coordinate[yshift = 3pt] (use) at (u.south east);
	\path (use) edge[out = 0, in = 105] node[below] {4} (p2);
	\coordinate[yshift = -3pt] (une) at (u.north east);
	\path (p2) edge[out = 90, in = 0] node[above] {2} (une);
	
	\node [transition, below of = tco, node distance = 1.3cm] (v) {$v$};
	\coordinate[yshift = -3pt] (vnw) at (v.north west);
	\path (p1) edge[out = 285, in = 180] (vnw);
	\coordinate[yshift = 3pt] (vsw) at (v.south west);
	\path (vsw) edge[out = 180, in = 270] node[below] {2} (p1);
	\path (v) edge[out = 0, in = 270] (p2);
\end{tikzpicture}
			\captionof{figure}{\label{Fig:Runexample} Petri net with places $p_1, p_2$, transitions $u,t,v$. 
				Edges are entries of the flow function $F$. 
				We omit the label if it is $1$.
			}
		\end{minipage}
		\hspace{0.3cm}
		\begin{minipage}{.43\textwidth}
			\centering
			\begin{tikzpicture}[
	dot/.style = {circle, fill = black, inner sep = 0pt, minimum height = 2pt},
	reddot/.style = {circle, fill = blue, inner sep = 0pt, minimum height = 2pt},
	greendot/.style = {circle, fill = forestgreen, inner sep = 0pt, minimum height = 2pt}
	]
	
	% Origin and axis
	\coordinate (origin) at (0,0);
	
	% tMinus and borders of Act(t)
	\coordinate (tMinus) at (1,0.5);
	\node[yshift = -5pt] at (tMinus) {\scriptsize $\tMinus$};
	\draw[-] (tMinus) -- ++(1.5,0);
	\draw[-] (tMinus) -- ++(0,2);
	
	% Dots
	\foreach \x in {0,1,...,5}{
		\foreach \y in {0,1,...,5}{
			\node[dot] at ($(origin) + (0.5*\x,0.5*\y)$) {};
		}
	}
	
	% Activation space
	\fill [yellow, opacity = 0.3] (tMinus) rectangle (2.6,2.6);
	
	% Invariant and half space
	\draw[-,shorten >= -3pt, shorten <= -3pt] (1.5,0) -- (0,2.25);
	\foreach \x in {3,...,5}{
		\foreach \y in {1,...,5}{
			\node[reddot] at ($(origin) + (0.5*\x,0.5*\y)$) {};
		}
	}
	\foreach \y in {2,...,5}{
		\node[reddot] at ($(origin) + (1,0.5*\y)$) {};
	}
	
	% M0 and Mf
	\node[greendot] (M0) at (1.5,0.5) {};
	\node[yshift = -7pt, xshift = 3pt] at (M0) {\scriptsize {\color{forestgreen}$\amark_0$}};
	\node[greendot,fill = red] (Mf) at (0,2) {};
	\node[yshift = -7pt] at (Mf) {\scriptsize {\color{red}$\amark_f$}};
	
	% tDelta
	\path[draw, ->, red, shorten >= 1pt] (1,1.5) --
	node[pos = 0.4, above, yshift = 2pt] {\scriptsize $\tDelta$}
	(0.5,2);
	
	% Invisible node for caption matching
	\node [below = 0.7cm of M0, opacity = 0] {$e$};
\end{tikzpicture}
			\captionof{figure}{\label{Fig:Rungeometry} Geometric interpretation of the half space $(k,c)$ in $\Z^2$.
				It is inductive and separates $m_0$ from $m_f$.}
		\end{minipage}
\end{figure}
%---------------------------------------------------------------------
%---------------------------------------------------------------------
%---------------------------------------------------------------------
%---------------------------------------------------------------------
\paragraph{Example}
We provide some geometric intuition.
Consider the Petri net in \autoref{Fig:Runexample}. 
Focus on transition $t$. 
The vectors describing $t$ are \mbox{$\tMinus = (2,1)$} (incoming edges), \mbox{$\tPlus = (1,2)$}  (outgoing edges), and $\tDelta = (-1,1)$.
The activation space of $t$ is \mbox{$\Act{t} = \setcond{(2,1) + (x,y)}{x,y \in \N}$}.
It is visualized by the yellow area in \autoref{Fig:Rungeometry}.
Let {\color{forestgreen}$\amark_0 = (3,1)$} and {\color{red}$\amark_f = (0,4)$}.
Consider the half space defined by $k = (3,2)$ and $c = 9$.
In \autoref{Fig:Rungeometry}, it is indicated by the diagonal line $k \cdot x = c, x \in \Real^2$. 
The set of integer vectors above it is $\Sol{k}{c}$.
Clearly, $\amark_0 \in \Sol{k}{c}$ and $\amark_f \notin \Sol{k}{c}$, the half space separates the markings.
The markings in {\color{blue}$\Act{t} \cap \Sol{k}{c}$} are colored blue in \autoref{Fig:Rungeometry}.
The half space is $t$-inductive:
if $\amark \in \Act{t} \cap \Sol{k}{c}$, firing $t$ does not lead to a marking below the line.
As we will see in \autoref{Section:Charinduct}, $(k,c)$ is also $u$ and $v$-inductive. 
Hence, it proves non-reachability of $\amark_f$ from $\amark_0$.
%---------------------------------------------------------------------
%---------------------------------------------------------------------
%---------------------------------------------------------------------
%---------------------------------------------------------------------
\paragraph{Linear Safety Verification}
Our goal is to find inductive half spaces that disprove reachability or coverability.
Given a Petri net $N$ and two markings $\amark_0,\amark_f$, we study two corresponding algorithmic problems:
the \emph{linear safety verification problem} \probnamereach~for reachability and its coverability variant \probnamecov.
%---------------------------------------------------------------------
%---------------------------------------------------------------------
%---------------------------------------------------------------------
%---------------------------------------------------------------------
\begin{description}
	\item[\probnamereach:] {} Is there an IHS $(k,c)$ with $\amark_0 \in \Sol{k}{c}$ and $\amark_f \notin \Sol{k}{c}$? 
	\item[\probnamecov:]{} Is there an IHS $(k,c)$ with $\amark_0 \in \Sol{k}{c}$ and ${\uparrow\!\amark_f} \cap \Sol{k}{c}=\emptyset$? 
\end{description}
%---------------------------------------------------------------------
%---------------------------------------------------------------------
%---------------------------------------------------------------------
%---------------------------------------------------------------------

%	Formally, the first algorithmic problem we will study is \probnamereach, 
%\emph{linear safety verification (reachability version)}: 
%\begin{quote}
%	{\bfseries Given:} A Petri net $N$ and two markings $\amark_0,\amark_f$.\\
%	{\bfseries Question:} Is there an IHS $(k,c)$ such that $\amark_0 \in \Sol{k}{c}$ and $\amark_f \notin \Sol{k}{c}$? 
%\end{quote}
%Problem \probnamecov, \emph{linear safety verification (coverability version)}, is defined by: 
%\begin{quote}
%	{\bfseries Given:} A Petri net $N$ and two markings $\amark_0,\amark_f$.\\
%	{\bfseries Question:} Is there an IHS $(k,c)$ such that $\amark_0 \in \Sol{k}{c}$ and ${\uparrow\!\amark_f} \cap \Sol{k}{c}=\emptyset$? 
%\end{quote}
%---------------------------------------------------------------------
%---------------------------------------------------------------------
%---------------------------------------------------------------------
%---------------------------------------------------------------------
The reader familiar with separability will note that disproving reachability of $\amark_f$ from~$\amark_0$ amounts to finding a separator between $\mathit{post}^*(\amark_0)$ and $\mathit{pre}^*(\amark_f)$.  
A separator is a set $S\subseteq \N^n$ so that $\mathit{post}^*(\amark_0)\subseteq S$ and $S\cap \mathit{pre}^*(\amark_f)=\emptyset$. 
The difference between separability and linear safety verification is that separators are neither required to be half spaces nor required to be inductive.

The choice for half spaces and inductivity is motivated by the constraint-based approach to safety verification that we pursue. 
Half spaces can be given in terms of $(k, c)$, a format that is computable by a solver.
Inductivity yields a local check for separation.
Indeed, if $(k, c)$ is inductive and $\amark_0 \in \Sol{k}{c}$, we already have $\mathit{post}^*(\amark_0)\subseteq \Sol{k}{c}$.
Similarly, if $(k, c)$ is inductive and $\amark_f \notin \Sol{k}{c}$, then $\mathit{pre}^*(\amark_f)\cap \Sol{k}{c}=\emptyset$.
Hence, $\Sol{k}{c}$ is indeed a separator.
But there are separators that are neither half spaces nor inductive.
To see the latter, consider a transition that is not enabled in $\mathit{post}^*(\amark_0)$ but in a separator $S$.
Firing the transition may lead to a marking outside of $S$ and violate inductivity.

While reachability and coverability are decidable for Petri nets, decidability of \probnamereach~and \probnamecov~is unknown.
Our approach semi-decides both problems.
% REMOVED:
%The difficulty is to synthesize half spaces that are indeed inductive.
%As we will see, the space of IHS decomposes into \emph{trivial} half spaces that are easy to identify as inductive and \emph{non-trivial} half spaces that need to be tested.
%Synthesizing the latter is hard.
%We establish a over-approximates this part.
%
%Once a full characterization of non-trivial IHS is found, a decision procedure is immediate.\flo{es gibt schon eine characterization, eine besser würde zu einer besseren decision procedure führen - falls es eine bessere gibt. das in die conclusion?}
%We tackle the problem by decomposing the set of half spaces into two classes.
%So-called trivial half spaces will always be inductive (and so the check is trivial).
%Non-trivial half spaces are the ones for which inductivity is non-trivial to check. 
%Our main findings are necessary conditions for non-trivial half spaces to be inductive. %given in \autoref{sec:pn:nontrivial}.

\section{Half Spaces}
\label{Section:Charinduct}
In order to synthesize inductive half spaces, we consider the structure of the space of IHS in more detail.
Our goal is to derive a linear constraint system that closely approximates the structure of the space.
The system can then be passed to an SMT-solver to synthesize candidates of half spaces.

Since IHS require inductivity for all transitions, their structure can be convoluted.
Therefore, we do not immediately consider the space of all IHS.
Instead, we first focus on half spaces that are inductive for a single transition.
We derive linear constraints describing these half spaces. 
They are combined in \autoref{Section:Generating} in order to obtain the desired SMT-formula for the space of all IHS.

The set of half spaces that are inductive for a given transition splits into two parts: the \emph{trivial} half spaces and the \emph{non-trivial} ones.
We first focus on the former.
Trivial half spaces were already described in~\cite{Sankaranarayanan2003,Triebel2015}.
They satisfy one of three conditions that immediately imply inductivity and can be easily synthesized.
We provide a formal definition below.
%---------------------------------------------------------------------
%---------------------------------------------------------------------
%---------------------------------------------------------------------
%---------------------------------------------------------------------

The first condition for triviality describes the fact that the vector $k$ and the transition $t$ point into the same direction.
The half space $(k,c)$ is \emph{oriented towards transition $t$} if $k \cdot \tDelta \geq 0$.
Since the scalar product provides information about the angle between $k$ and $\tDelta$, the condition means that firing transition $t$ moves a marking in the half space further away from the border.
To give an example, consider the half space $(k,c)$ with $k = (3,2)$ from~\autoref{Fig:Rungeometry}. 
It is oriented towards transitions $u$ and $v$.
We have $u^\Delta = (-1,2)$ and $v^\Delta = (1,1)$, hence $k \cdot u^\Delta$ and $k \cdot v^\Delta$ are both non-negative.
The half space is not oriented towards $t$ since $t^\Delta = (-1,1)$.
Firing $t$ means moving closer to the border of the half space.

It easy to see that a half space which is oriented towards a transition $t$ is actually $t$-inductive.
This observation is a first step in the synthesis of IHS.
In~fact, note that generating a half space $(k,c)$ that separates two markings $\amark_0$ and $\amark_f$ and that is oriented towards $t$ amounts to finding a solution $(k,c)$ of the linear constraint system
$k \cdot \amark_0 \geq c \wedge k \cdot \amark_f < c \wedge k \cdot \tDelta \geq 0$.

The second condition for triviality uses the fact that for $k \geq 0$, the function $k \cdot \amark$ is monotone on markings.
We call a half space $(k,c)$ \emph{monotone for transition $t$} if $k \geq 0$ and $k \cdot (\tMinus + \tDelta) \geq c$.
Note that with larger markings, $k \cdot m$ grows.
This means if the smallest marking in the half space enabling $t$, namely $\tMinus$, stays within the half space after firing $t$, the same holds for all larger markings.
The requirement is captured in the inequality $k \cdot (\tMinus + \tDelta) \geq c$.
Hence, monotone half spaces are inductive and can be synthesized as solutions of $k \geq 0 \wedge k \cdot (\tMinus + \tDelta) \geq c$.

The last condition is dual to monotonicity.
A half space $(k,c)$ is \emph{antitone for transition $t$} if $k \leq 0$ and $k \cdot \tMinus < c$.
The latter requirement describes that $\tMinus$ does not lie in the half space.
Since $k \leq 0$ this means that $\Act{t} \cap \Sol{k}{c} = \emptyset$.
Hence, antitone half spaces are inductive.
Moreover, they can be generated as solutions to the linear constraints $k \leq 0 \wedge k \cdot \tMinus < c$.
We summarize:
%---------------------------------------------------------------------
%---------------------------------------------------------------------
%---------------------------------------------------------------------
%---------------------------------------------------------------------
\begin{definition}
	A half space $(k,c)$ is \emph{trivial wrt. $t$} if one of the following holds: $(k,c)$ is oriented towards $t$, $(k,c)$ is monotone for $t$, or $(k,c)$ is antitone for $t$.
\end{definition}
%---------------------------------------------------------------------
%---------------------------------------------------------------------
%---------------------------------------------------------------------
%---------------------------------------------------------------------
\begin{theorem}{(\cite{Triebel2015})}\label{Theorem:TrivialInductive}
	If $(k,c)$ is trivial with respect to $t$ then it is $t$-inductive.
\end{theorem}
%---------------------------------------------------------------------
%---------------------------------------------------------------------
%---------------------------------------------------------------------
%---------------------------------------------------------------------

Non-trivial half spaces are not automatically $t$-inductive.
As an example, consider the half space from \autoref{Fig:Rungeometry}.
Recall that $k = (3,2)$ and $c = 9$.
If we replace $c$ by $c' = 8$, we get that $(k,c')$ is a non-trivial half space that is not $t$-inductive.
We have $k \cdot \tMinus = 8 = c'$ but $k \cdot (\tMinus + \tDelta) = 7 < c'$.
Hence, when firing $t$ from $\tMinus$, we leave $(k,c')$.
This has two implications.
First, we need an algorithm to test whether a non-trivial half space is indeed inductive.
Second, we cannot hope for a simple synthesis as for trivial half spaces.
The former is resolved by the algorithm \ic~which we show in \autoref{Sec:checking}.
For the latter, we develop an independent structure theory in the subsequent section.
\subsection{Non-Trivial Half Spaces}\label{sec:pn:nontrivial}
We consider half-spaces that are non-trivial but inductive.
These are neither oriented towards the transition of interest, nor monotone, nor antitone.
Our first insight is a structural theorem which strongly impacts the synthesis of non-trivial IHS.
In fact, we show that a half space $(k,c)$ which is not oriented towards a transition $t$ but $t$-inductive cannot have positive and negative entries in $k$.
This means we can restrict to $k \geq 0$ or $k \leq 0$ when synthesizing non-trivial IHS.
%By analyzing their structure, we comprehend how restrictive the assumption of triviality is.
%The first concern is this.
%If a half space $(k,c)$ is not oriented towards a transition, triviality requires $k \geq 0$ or $k \leq 0$.
%This means we could be missing IHS where $k$ contains positive and negative entries.
%The main finding is that such IHS cannot exist.
%Instead, triviality is complete in the following sense.
%---------------------------------------------------------------------
%---------------------------------------------------------------------
%---------------------------------------------------------------------
%---------------------------------------------------------------------
\begin{restatable}{theorem}{unmixed}
	\label{Theorem:Unmixed}
	Let $(k, c)$ be a half space that not oriented towards a transition $t$ but $t$-inductive.
	Then, we have $k\geq 0$ or $k\leq 0$.
\end{restatable}
%---------------------------------------------------------------------
%---------------------------------------------------------------------
%---------------------------------------------------------------------
%---------------------------------------------------------------------
The proof of the theorem relies on the notion of \emph{syzygies} known from commutative algebra~\cite{Greuel2002}.
We adapt it to our setting.
A \emph{syzygy of $k$} is a vector $s \in \Z^n$ with $k \cdot s = 0$.
This means that adding a syzygy to a marking $m$ does not change the scalar product with $k$.
We have $k \cdot m = k \cdot (m + s)$.
Hence, if $m \in \Sol{k}{c}$, we get that $m + s \in \Sol{k}{c}$ for all syzygies $s$ of $k$.
We proceed with the proof.
%	Given a finitely generated module $M$ over a ring R and a set $k_1,...,k_n$ of generators, a syzygy of $M$ is an element $(s_1,...,s_n) \in R^n$ for which $s_1 k_1+\cdots +s_n k_n=0$ holds.
%	For our purpose, the intuition behind a syzygy is that it is a vector that returns 0 if it is multiplied with $k$. 
% -------------------------------------------------------------
% -------------------------------------------------------------
% -------------------------------------------------------------
% -------------------------------------------------------------
\begin{proof}
	Assume $(k,c)$ is $t$-inductive and not oriented towards $t$ but there are $i \neq j$ with $k(i) > 0$ and $k(j) < 0$.
	We show that $(k,c)$ cannot be $t$-inductive which contradicts the assumption.
	The idea is as follows.
	 We set $\avecp(i) = \lceil \frac{c}{k(i)} \rceil$ and $\avecp(\ell) = 0$ for $\ell \neq i$. 
	 Note that $\avecp \in \Sol{k}{c}$.
	From $u$, we construct a vector $v \in \Z^n$ that lies in $\Sol{k}{c}$ but $v + \tDelta \notin \Sol{k}{c}$.
	Note that $v$ might not be a proper marking.
	By adding non-negative syzygies to $v$, we obtain a marking $m \in \Act{t} \cap \Sol{k}{c}$ with $m + \tDelta \notin \Sol{k}{c}$.
	Hence, $(k,c)$ is not $t$-inductive.
	%We show that any half space $(k, c)$ that is not oriented towards $t$ and where $k(i)>0$ and $k(j)<0$ for some $i, j\in [1, n]$ cannot be $t$-inductive.
%	Let $(k,c)$ be such an inequality and $\avecp \in \Z^n$ a vector in $\Sol{k}{c}$. 
%	Note that such a vector always exists.
%	We construct from $\avecp$ a vector $\avec$ close to the border of the half space.
%	Then the idea is to lift $\avec$ to a marking $\amark$ by adding non-negative syzygies. 
%	The resulting marking lies in $\Act{t}\cap\Sol{k}{c}$ and shows that $(k,c)$ cannot be $t$-inductive by satisfying $\amark +\tDelta\not\in \Sol{k}{c}$.
		
	The vector $\avec$ is defined by $\avec = \avecp + \lfloor \frac{c - k\cdot \avecp} {k \cdot \tDelta} \rfloor \mal \tDelta \in \Z^n$.
	Since $(k,c)$ is not oriented towards $t$, we have $k \cdot \tDelta < 0$.
	Hence, $\avec$ is well-defined.
	By $\lfloor x \rfloor \geq x - 1$, we obtain the following inequality showing that $\avec \in \Sol{k}{c}$:
	\begin{align*}
		k \cdot \avec
		\geq k \cdot \avecp + \big( \frac{c - k \cdot \avecp}{k \cdot \tDelta} - 1\big) \cdot k \cdot \tDelta
		= c - k \cdot \tDelta 
		\geq c.
	\end{align*}
	
	Similarly, by $\lfloor x \rfloor \leq x$, we obtain that $v + \tDelta \notin \Sol{k}{c}$:
	\begin{align*}
		k \cdot (\avec + \tDelta)
		\leq k \cdot \avecp + \frac{c - k \cdot \avecp}{k \cdot \tDelta} \cdot k \cdot \tDelta + k \cdot \tDelta
		= c + k \cdot \tDelta < c.
	\end{align*}
	
	Note that $v$ is not yet a counter example for $t$-inductivity.
	Indeed, we cannot ensure that $\avec$ is a marking that enables $t$.
	But we can construct such a marking by adding syzygies to $\avec$.
	For a place $p \in P$ let $e_p$ denote the $p$-th unit vector.
	This means $e_p(p) = 1$ and $e_p(q) = 0$ for $q \neq p$.
	For any place $p$, we construct a syzygy $s_p$ defined as follows.
	If $k(p) > 0$, we set $s_p = -k(j)\mal e_p + k(p)\mal e_j$.
	If $k(p) < 0$, we set $s_p = -k(p)\mal e_i + k(i)\mal e_p$.
	For the case $k(p) = 0$, we simply set $s_p = e_p$.
	Note that for all places $p$, we have $s_p \geq 0$ and $k \cdot s_p = 0$.

	The syzygies $s_p$ allow for adding non-negative values to each component of $\avec$ without changing the scalar product with $k$.
	Hence, there exist $\mu_p \in \N$ such that $\avec + \sum_{p \in P} \mu_p \mal s_p \geq \tMinus$.
	By setting $m = \avec + \sum_{p \in P} \mu_p \mal s_p$, we get a marking in $\ACT(t)$ that satisfies $k \cdot m = k \cdot \avec \geq c$ and $k \cdot (m + \tDelta) = k \cdot (\avec + \tDelta) < c$.
	Hence, $m$ contradicts $t$-inductivity of $(k,c)$ and we obtain the desired contradiction.
	\qed
\end{proof}
%---------------------------------------------------------------------
%---------------------------------------------------------------------
%---------------------------------------------------------------------
%---------------------------------------------------------------------

%\input{relevant}

The theorem allows us to assume $k \geq 0$ or $k \leq 0$ when synthesizing non-trivial inductive half spaces.
However, we cannot hope for a compact linear constraint system like we have for trivial half spaces.
The reason is as follows.
Assume we have a constraint system $L(k,c)$ of polynomial size describing the space of $t$-inductive non-trivial half spaces.
Each solution of $L(k,c)$ corresponds to such a half space and vice versa.
We can then decide, in polynomial time, whether a given half space $(k,c)$ is $t$-inductive.
Indeed, an algorithm would first decide whether $(k,c)$ is trivial or non-trivial.
In the former case, $t$-inductivity immediately follows.
In the latter case, the algorithm checks if $(k,c)$ is a solution to $L(k,c)$.
All these steps can clearly be carried out in polynomial time.
However, the algorithm would contradict the \coNP-hardness of checking $t$-inductivity, which we prove in \autoref{Sec:checking}.
Hence, the system $L(k,c)$ of polynomial size cannot exist.

%Non-trivial inequalities do not admit a simple structure like trivial inequalities do.
%In fact, we cannot even hope for a linear constraint system of polynomial size, describing the space of $t$-inductive non-trivial inequalities.
%The reason is as follows: if such a system would exist, call it $L(k,c)$, we could decide $t$-inductivity for each given \pinequ~$(k,c)$ in polynomial time.
%An algorithm for the problem would first test whether $(k,c)$ is trivial or non-trivial.
%In the first case, $t$-inductivity would be immediate.
%In the latter case, the algorithm would plug $(k,c)$ into the system $L(k,c)$ and evaluate.
%Depending on the evaluation, $(k,c)$ would be $t$-inductive or not.
%All the described steps can be carried out in polynomial time.
%Hence, this contradicts the $\NP$-hardness of the problem (in binary encoded input) which we will prove in \autoref{Sec:checking}.

Although a concise constraint system for the space of non-trivial IHS seems out of reach, we can give a close linear approximation.
To this end, we derive two necessary conditions for non-trivial IHS that can be formulated in terms of linear constraints.
The first one is given in the following lemma.
The proof follows from \autoref{Theorem:Unmixed} and from inverting the constraints for trivial half spaces.
%By utilizing the derived structural properties, we obtain a linear constraint system the solutions of which are precisely the vectors $k$ that appear in $t$-inductive non-relevant inequalities.
%The constraint system will be of particular importance when synthesizing inductive invariants.
%---------------------------------------------------------------------
%---------------------------------------------------------------------
%---------------------------------------------------------------------
%---------------------------------------------------------------------
\begin{lemma}\label{Lemma:Nontrivial}
	A $t$-inductive half space $(k,c)$ that is non-trivial for $t$ either satisfies (a) $k \geq 0$ and $k \cdot \tMinus < c - k \cdot \tDelta$ or (b) $k \leq 0$ and $k \cdot \tMinus \geq c$.
\end{lemma}
%---------------------------------------------------------------------
%---------------------------------------------------------------------
%---------------------------------------------------------------------
%---------------------------------------------------------------------

%Note that \autoref{Theorem:Unmixed} is essential for proving the lemma.
%The statement does not follow immediately from negating the definition of trivial inequalities.
The lemma provides geometric intuition to separate non-trivial from trivial half spaces.
If $(k,c)$ is non-trivial, $\Sol{k}{t} \cap \Act{t}$ is a strict non-empty subset of the activation space $\Act{t}$.
This stands in contrast to the trivial case.
Here, $(k,c)$ is either oriented towards $t$ or the following holds.
If $(k,c)$ is monotone, we have $\Sol{k}{t} \cap \Act{t} = \Act{t}$ 
and if $(k,c)$ is antitone, we have $\Sol{k}{t} \cap \Act{t} = \emptyset$.

We employ \autoref{Lemma:Nontrivial} to derive a further necessary condition for non-trivial half spaces.
It provides a lower bound for the absolute values of the vector $k$.
%To this end, we invoke the criterion of inductivity given in \autoref{Theorem:inductive}.
% where  $\Sol{k}{t} \cap \Act{t}$ is either empty or (nearly) the whole activation space.
%Moreover, \autoref{Lemma:Nontrivial} is a tool for proving the following necessary criterion.
%---------------------------------------------------------------------
%---------------------------------------------------------------------
%---------------------------------------------------------------------
%---------------------------------------------------------------------
\begin{restatable}{lemma}{NecessNontrivial}\label{Lemma:NecessNontrivial}
	Let $(k,c)$ be a $t$-inductive half space that is non-trivial for $t$.
	For any entry $k(i)$ of $k$, with $|k(i)|$ denoting its absolute value, we have:
	%Then, for any place $p$, it holds 
	\begin{align}
	k(i) = 0\; \lor \; \abs{k(i)} \geq - k \cdot \tDelta  \tag{5}
	\end{align}
\end{restatable}
%---------------------------------------------------------------------
%---------------------------------------------------------------------
%---------------------------------------------------------------------
%---------------------------------------------------------------------

The idea behind the lemma is the following.
If the absolute value of an entry of $k$ is too small then we can construct a vector $x \in \N^n$ such that $k \cdot x + k \cdot \tMinus$ lies between $c$ and $c - k \cdot \tDelta-1$.
This violates the condition stated in \autoref{Theorem:inductive}. 

\section{Generating Invariants}
\label{Section:Generating}
We combine the conditions from \autoref{Section:Charinduct} to formulate a linear SMT-formula $\phi$ approximating the space of inductive half spaces.
A solution to $\phi$ is a vector $k$ that potentially forms an IHS.
To keep the constraints in the formula linear, we cannot generate a corresponding constant $c$ immediately.
Instead, we replace $c$ by bounds imposed by \probnamereach~and \probnamecov~and generate candidates for $c$ in a second synthesis step with the algorithm \cg.
The algorithm is given in \autoref{Sec:checking}.

Recall that in \probnamereach, we are interested in finding an inductive half space $(k,c)$ that separates an initial marking $\amark_0$ from a marking $\amark_f$.
Phrased differently, we want $\amark_0 \in \Sol{k}{c}$ and $\amark_f \notin \Sol{k}{c}$.
The former implies that $k \cdot \amark_0 \geq c$, the latter implies $k \cdot \amark_f < c$.
The inequalities yield that $k \cdot \amark_0 > k \cdot \amark_f$ and impose two bounds on $c$, namely $c \in [k \cdot \amark_f + 1, k \cdot \amark_0]$.
We apply the bounds to the constraints obtained for trivial half spaces and derive the following conditions:
\begin{minipage}{.4\textwidth}
	\begin{align}
		&k\mal \amark_0>k\mal \amark_f \tag{0}\label{eq:gen0}\\
		&k \mal t^\Delta \geq 0 \tag{1}\label{eq:gen1}
	\end{align}
\end{minipage}
\begin{minipage}{.5\textwidth}
	\begin{align}
		&k\leq 0 \; \land \; k\mal t^- < k\mal m_0 \tag{2}\label{eq:gen2}\\
		&k\geq 0 \; \land \; k\mal t^- > k\mal m_f -k\mal \trans \tag{3}\label{eq:gen3}
	\end{align}
\end{minipage}
\medskip
%---------------------------------------------------------------------
%---------------------------------------------------------------------
%---------------------------------------------------------------------
%---------------------------------------------------------------------
%				\begin{align}
%				&k\mal \amark_0>k\mal \amark_f \tag{0}\label{eq:gen0}\\
%&k \mal t^\Delta \geq 0 \tag{1}\label{eq:gen1}\\
%&k\leq 0 \; \land \; k\mal t^- < k\mal m_0 \tag{2}\label{eq:gen2}\\
%&k\geq 0 \; \land \; k\mal t^- > k\mal m_f -k\mal \trans \tag{3}\label{eq:gen3}
%				\end{align}

Each of the conditions (\ref{eq:gen1}), (\ref{eq:gen2}), and (\ref{eq:gen3}) models a type of trivial half spaces.
For instance, (\ref{eq:gen1}) describes half spaces that are oriented towards $t$.
Together with (\ref{eq:gen0}), we ensure $t$-inductivity for some $c$ within the bounds.
To describe non-trivial half spaces, we employ~\autoref{Theorem:Unmixed} and \autoref{Lemma:NecessNontrivial}. 
We derive the following constraints: \newline
%Each single condition (\ref{eq:gen1})-(\ref{eq:gen3}) together with (\ref{eq:gen0}) ensures trivial t-inductivity for some values of $c$ within the limits.
%If none of these condition hold, then we get non-trivial inequalities. In this case ~\autoref{Theorem:Unmixed} and \autoref{Lemma:NecessNontrivial} apply and we get \ref{eq:gen4} and \ref{eq:gen5} respectively:
%---------------------------------------------------------------------
%---------------------------------------------------------------------
%---------------------------------------------------------------------
%---------------------------------------------------------------------
\begin{minipage}{.4\textwidth}
	\begin{align}
		k\geq 0 \lor k\leq 0 \tag{4}\label{eq:gen4}
	\end{align}
	\end{minipage}
	\begin{minipage}{.5\textwidth}
	\begin{align}
		\forall_i ~ k(i) = 0\; \lor \; \abs{k(i)} \geq - k \cdot \tDelta  \tag{5}\label{eq:gen5}
	\end{align}
\end{minipage}
\medskip
%---------------------------------------------------------------------
%---------------------------------------------------------------------
%---------------------------------------------------------------------
%---------------------------------------------------------------------

We collect all the constraints in the SMT-formula $\closedmt$ in order to find the desired inductive half space.
Note that the half space must be separating (\ref{eq:gen0}).
Moreover, it is either trivial, so it satisfies one out of (\ref{eq:gen1}), (\ref{eq:gen2}), and (\ref{eq:gen3}), or it is non-trivial and satisfies (\ref{eq:gen4}) and (\ref{eq:gen5}).
We construct the formula accordingly:
\begin{align*}
	\closedmt \define (\ref{eq:gen0}) \land ((\ref{eq:gen1})\lor (\ref{eq:gen2}) \lor (\ref{eq:gen3}) \lor ((\ref{eq:gen4}) \land (\ref{eq:gen5}))).
\end{align*}

As mentioned above, it is not possible to construct a linear constraint system of polynomial size that captures all $t$-inductive half spaces and yields $k$ and $c$.
However, $\closedmt$ is a tight approximation.
In fact, its solutions are precisely those vectors $k$ that can form a $t$-inductive half space which separates $\amark_0$ from $\amark_f$.
%---------------------------------------------------------------------
%---------------------------------------------------------------------
%---------------------------------------------------------------------
%---------------------------------------------------------------------
\begin{restatable}{lemma}{genc}
	\label{lem:genc}
	There exists a constant $c\in \Z$ such that $(k,c)$ is a $t$-inductive half space with $k \cdot \amark_0 \geq c$ and $k \cdot \amark_f < c$ if and only if $k$ is a solution to $\closedmt$.
\end{restatable}
%---------------------------------------------------------------------
%---------------------------------------------------------------------
%---------------------------------------------------------------------
%---------------------------------------------------------------------

Our goal is to synthesize an IHS that separates $\amark_0$ from $\amark_f$.
Since IHS are $t$-inductive for all transitions $t$, we join all $\closedmt$ in a conjunction $\closedm \define \bigwedge_{t \in T} \closedmt$.
The SMT-formula $\closedm$ describes the desired linear approximation of the space of IHS.
It is a main ingredient of our CEGAR loop outlined in~\autoref{Figure:cegar}.
According to \autoref{lem:genc}, solutions to $\closedm$ are those vectors $k$ that admit a constant $c_t$ for each transition $t$ such that $(k,c_t)$ is $t$-inductive.
The problem is that these $c_t$ may be different for each transition.
Hence, $\closedm$ generates half space candidates and what is left to find is a single value $c$ such that $(k,c)$ is $t$-inductive for each $t$.
We can compute all possible values for $c$ with the algorithm \cg.
A detailed explanation is given in \autoref{Sec:checking}.
Once a common $c$ is found, we have synthesized the desired IHS.
Otherwise, the CEGAR loop starts the refinement.
If a solution $k$ of $\closedm$~does not have a suitable constant $c$ to form an IHS, then neither does any multiple of $k$.
This means we can exclude all multiples  in future iterations of the CEGAR loop.
Let $\mul(k)$ be the formula satisfied by a $k' \in \Z^n$ if and only if there exists an $a \in \N$ such that $a \cdot k = k'$.
Then, the refinement performs the update $\closedm \define \closedm \land \neg \mul(k)$.
The following lemma states correctness.
%---------------------------------------------------------------------
%---------------------------------------------------------------------
%---------------------------------------------------------------------
%---------------------------------------------------------------------
\begin{restatable}{lemma}{lemmamult}
	Let $k'\define a\mal k$ with $a\in \N$. If $(k',c)$ is an IHS, then so is $(k, \lceil \frac c a \rceil )$.
\end{restatable}
%---------------------------------------------------------------------
%---------------------------------------------------------------------
%---------------------------------------------------------------------
%---------------------------------------------------------------------
%\begin{proof}
%	We use contraposition. 
%	Given a marking $\amark$ that violates t-inductivity of $(k, \lceil \frac c a \rceil )$. 
%	We show that it violates t-inductivity of $(k',c)$ as well.
%	
%	If it holds $k\mal \amark \geq \lceil \frac c a \rceil$, then it follows $ k\mal \amark \geq \frac c a$ from $\lceil \frac c a \rceil\geq \frac c a$ and thus $ k'\mal \amark \geq c$.
%	 %Note that $\frac{c}{a}$ may not be an integer but $k\mal \amark$ is an integer that is at least  $\frac{c}{a}$ and thus at least $\lceil \frac c a \rceil$. This means  $(k, \lceil \frac c a \rceil )$ holds. 
%	If it holds  $k\mal (\amark +\tDelta)<\lceil \frac c a \rceil $ then it follows 
%	 $k\mal (\amark +\tDelta)+1 \leq \lceil \frac c a \rceil $. Finally, we conclude  
%	 $k'\mal (\amark +\tDelta)<c$ using $\lceil \frac c a \rceil \leq \frac c a +1$.\qed
%\end{proof}
%This lemma shows that our method of excluding multiples does not forbid vectors that may form an IHS. If a multiple $k'=a\mal k$ of $k$ could form an IHS $(k',c)$, then $k$ would have already formed an IHS $(k, \lceil \frac c a \rceil )$.

The presented CEGAR approach generates inductive half spaces.
In order to semi-decide \probnamereach, our approach needs to yield an IHS whenever we are given a yes-instance.
This means we need to ensure that any candidate vector $k$ is generated by the SMT-solver at some point so that we do not miss possible IHS.
This is achieved by adding a constraint imposing a bound on the absolute values of the entries of $k$.
If the formula becomes unsatisfiable, the bound is increased.
It remains to show how our semi-decider for \probnamereach~can be adapted to \probnamecov.
%---------------------------------------------------------------------
%---------------------------------------------------------------------
%---------------------------------------------------------------------
%---------------------------------------------------------------------
\paragraph{Coverability}
Recall that a solution $k$ of $\closedm$ satisfies Condition (\ref{eq:gen0}).
It ensures the existence of a value $c$ such that $k \cdot \amark_0 \geq c$ and $k \cdot \amark_f < c$, meaning $\amark \in \Sol{k}{c}$ and $\amark_f \notin \Sol{k}{c}$.
While this is sufficient for disproving reachability, it is not for coverability.
When we solve \probnamecov, we need to additionally guarantee that $\uparrow\!\amark_f \cap \Sol{k}{c}$ is empty.
It turns out that this requirement can be captured by a simple modification of $\closedm$. We only need to ensure that $k$ is negative.
%---------------------------------------------------------------------
%---------------------------------------------------------------------
%---------------------------------------------------------------------
%---------------------------------------------------------------------
\begin{theorem}\label{Theorem:CoverNegative}
	Let $(k,c)$ be a half space (not necessarily inductive) such that $\amark_f\notin \Sol{k}{c}$.
	Then we have $\uparrow\!\amark_f \cap \Sol{k}{c} =\emptyset$ if and only if $k \leq 0$.
\end{theorem}
%---------------------------------------------------------------------
%---------------------------------------------------------------------
%---------------------------------------------------------------------
%---------------------------------------------------------------------
The intuition is as follows. 
If $k\leq 0$ does not hold, then we can start with $\amark \define \amark_f$ and put tokens into a place $i$ with $k_i>0$ until $k \mal \amark\geq c$. This means $k\leq 0$ is sufficient and necessary.
Each solution $k$ of $\closedm$ satisfies $\amark_f\notin \Sol{k}{c}$ for some $c$.
In order to disprove coverability, we apply \autoref{Theorem:CoverNegative}~and add constraint $k \leq 0$ to $\closedm$.
This ensures that any synthesized IHS separates $\amark_0$ from $\uparrow\!\amark_f$.
\section{Checking Inductivity}
\label{Sec:checking}

We present the algorithms \ic~and \cg.
The former decides $t$-inductivity for a given half space $(k,c)$ and transition $t$.
The latter is an instrumentation of \ic~capable of synthesizing all constants $c$ such that $(k,c)$ is $t$-inductive, if only the vector $k$ is given.
%We examine the problem of deciding $t$-inductivity for a given half space $(k,c)$ and transition $t$.
%We design an algorithm that solves this.
%Moreover, we show how the algorithm can be instrumented to compute all $c$ such that $(k,c)$ is an IHS if only $k$ is given.
\cg~constitutes the remaining bit of our CEGAR loop.
Finally, we show that deciding $t$-inductivity is an \coNP-complete problem.
The proof once again employs a connection to discrete mathematics.
%---------------------------------------------------------------------
%---------------------------------------------------------------------
%---------------------------------------------------------------------
%---------------------------------------------------------------------
\subsection{Algorithms}
We start with the \emph{inductivity checker} (\ic).
Fix a half space $(k,c)$ and a transition $t$.
We need to decide whether $(k,c)$ is $t$-inductive.
If $(k,c)$ is trivial with respect to $t$, then inductivity follows from \autoref{Theorem:TrivialInductive}.
Hence, we assume that $(k,c)$ is non-trivial.
%For a given trivial half space $(k,c)$, $t$-inductivity immediately follows from \autoref{Theorem:TrivialInductive}.
%Hence, we may assume that the input $(k,c)$ is non-trivial.
The idea of \ic~is to algorithmically check the constraint formulated in \autoref{Theorem:inductive} via dynamic programming.
Roughly, the algorithm searches for a value $k \cdot \amark$, where $\amark \in \act{t}$, that lies in the \emph{target interval} $[c,c-k \cdot \tDelta - 1]$.
If such a value can be found, $(k,c)$ is not $t$-inductive.
Otherwise, it is $t$-inductive.

To state \ic, we adapt \autoref{Theorem:inductive}.
Let $K\define \setcond{k(i)}{i \in [1,n]}$ contain all entries of the the given vector $k$.
We consider sequences $\sumk{1} \ldots \sumk{\ell} \in K^*$.
Note that $\sumk{i}$ does not denote the $i$-th entry of $k$ but the $i$-th element in the sequence.
Then, $(k,c)$ is $t$-inductive if and only if there does not exist a sequence $\sumk{1} \ldots \sumk{\ell}$ with
\begin{align}
	c \leq k \cdot \tMinus + \sum_{i=1}^{\ell} \sumk{i} < c - k \cdot \tDelta. \tag{6} \label{eq:sumkthm}
\end{align}
%---------------------------------------------------------------------
%---------------------------------------------------------------------
%---------------------------------------------------------------------
%---------------------------------------------------------------------
\begin{algorithm}
	\caption{\emph{Inductivity Checker} (\ic)}
	\label{alg:Infinum}
	queue.add($ k \mal t^-$)\;
	reached[$ k \mal t^-$]$\define$ True\;
	\Repeat{queue.isEmpty}{
		$current$:=queue.remove()\;
		\If{$c\leq current <c-k\mal \trans $}{\label{alg:ifuninductive}
			\Return Not inductive\;
		}
		\For{$k\in K$}{
			\If{$(current +k <c-k\mal \trans \land k\geq 0)$ \label{alg:ifadda}\\
				$\lor (current +k \geq c \land k\leq 0) $} {	\label{alg:ifaddb}
				\If{$\neg$reached[$current+k$]}{
					queue.add($current+k$)\label{alg:queueadd}\;
					reached[$current+k$]:=True
				}
			}	
		}
	}
	\Return Inductive
\end{algorithm}
%---------------------------------------------------------------------
%---------------------------------------------------------------------
%---------------------------------------------------------------------
%---------------------------------------------------------------------

\ic~is stated as \autoref{alg:Infinum}.
It searches for a sequence in $K^*$ satisfying~(\ref{eq:sumkthm}).
Recall that we assumed $(k,c)$ to be non-trivial.
Then, according to \autoref{Theorem:Unmixed}, $k$ does not contain both, positive and negative entries.
\ic~starts at $k \cdot \tMinus$ and iteratively adds values of $K$ until it either reaches the target interval $[c,c-k \cdot \tDelta - 1]$ or finds that none such value is reachable.
To this end, \ic~employs dynamic programming.
This avoids recomputing the same value and speeds up the running time.
An example of a run of \ic~is illustrated in \autoref{fig:run}.
If the currently reached value lies below the target interval, at least one value of $K$ has yet to be added.
Once we overshoot the target interval, we can exclude the current value and go to the next one in the queue.
When we hit the interval, we can report non-inductivity.

In the appendix we show that \ic~is correct.
Moreover, we prove that it runs in pseudopolynomial time.
That is, polynomial in the values $k$,$c$, and $t$, or exponential in their bit size.
Note that this does not contradict the \coNP-hardness of checking $t$-inductivity which we prove below.
%---------------------------------------------------------------------
%---------------------------------------------------------------------
%---------------------------------------------------------------------
%---------------------------------------------------------------------
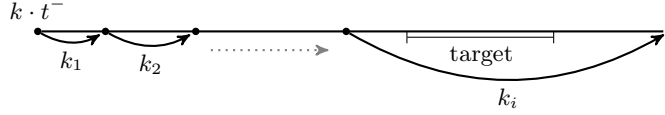
\begin{figure}
	\begin{center}
		\begin{tikzpicture}[->,>=stealth',shorten >=1pt,auto,node distance=1cm,
		thick,main/.style={draw,circle,fill=black,minimum size=2pt,
			inner sep=0pt}]
		
		\node[main, label=above:$k \cdot t^-$] (start) {};
		\node[main] (k1) [ right of=start,xshift=-0.1cm] {};
		\node[main] (k2) [ right of=k1,,xshift=0.2cm] {};
		\node[main] (kistart) [ right of=k2,xshift=1cm] {};
		\node[] (c) [ right of=kistart,xshift=-0.2cm] {};
		
		\node[] (ki2) [ right of=c] {};
		\node[] (c-kt-) [ right of=ki2] {};
		
		\node[] (ki1) [ right of=c-kt-,xshift=-0.5cm] {};
		\node[main] (ki) [ right of=ki1] {};
		
		\draw[-](start) to node[] {} (ki);
		\draw[->,bend right](start) to node[below] {$\sumk{1}$} (k1);
		\draw[->,bend right](k1) to node[below] {$\sumk{2}$} (k2);
		\draw[dotted, gray]([yshift=-2mm,xshift=2mm] k2.south) to node[] {} ([yshift=-2mm,xshift=-2mm] kistart.south);
		\draw[|-|,thin]([yshift=-2mm] c.north) to node[below] {target} ([yshift=-2mm] c-kt-.north);
		
		\draw[->,bend right](kistart) to node[below] {$\sumk{i}$} (ki);
		\end{tikzpicture}
	\end{center}
	\caption{Example run of \autoref{alg:Infinum} (\ic) in the case $k \geq 0$.
		It starts at value $k \cdot \tMinus$.
		The algorithm adds values of $K$ until it either overshoots the target interval or hits it.}
	\label{fig:run}
\end{figure}
\vspace*{-0.5cm}
%---------------------------------------------------------------------
%---------------------------------------------------------------------
%---------------------------------------------------------------------
%---------------------------------------------------------------------
\paragraph{Constant Generation}
Given a vector $k$ and a transition $t$, \ic~can be instrumented to compute all values $c$ such that $(k,c)$ is $t$-inductive.
We refer to the instrumentation as \emph{constant generation algorithm} (\cg).
\cg~computes all necessary sums $k \cdot \tMinus + \sum_{i=1}^\ell k_i$ with $k_1 \dots k_i \in K^*$ and returns all $c$ such that $[c,c-k \cdot \tDelta - 1]$ does not contain any of the computed sums.
Intuitively, we fit the interval between these sums.
Note that each of the returned values $c$ satisfies the characterization of $t$-inductivity as stated in~(\ref{eq:sumkthm}).
We show that \cg~is correct and terminates.

For termination, we need the so-called \emph{Frobenius number}~\cite{10.2307/2371684}.
Let $a \in \N^n$ be a vector such that $\gcd(a) = \gcd(a(1), \dots, a(n)) = 1$.
Here, $\gcd$ denotes the \emph{greatest common divisor}.
The \emph{Frobenius number of $a$} is the largest integer that cannot be represented as a positive linear combination of $a(1), \dots, a(n)$.
The number exists and is bounded by $a_{\max} \cdot a_{\min}$, where $a_{\max}$ is the largest and $a_{\min}$ the smallest entry of $a$~\cite{10.2307/2371684}.
Note that this means that each value $x \geq a_{\max} \cdot a_{\min}$ can be represented as a positive linear combination $x = a \cdot m$ with $m \in \N^n$.

This has implications for \cg.
Assume we are given a vector $k \geq 0$ with $\gcd(k) = 1$.
The possible values of $c$ such that $(k,c)$ is $t$-inductive cannot exceed $k \cdot \tMinus + k_{\max} \cdot k_{\min}$.
Otherwise, the interval $[c,c-k \cdot \tDelta - 1]$ will contain a linear combination of the form $k \cdot \tMinus + \sum_{i=1}^\ell k_i$ which breaks the inductivity requirement~(\ref{eq:sumkthm}). 
The argument can be generalized for any $gcd(k)\geq 1$:
%---------------------------------------------------------------------
%---------------------------------------------------------------------
%---------------------------------------------------------------------
%---------------------------------------------------------------------
\begin{restatable}{theorem}{boundonc}
	\label{thm:frobeniusbound}
	\label{corollary_bound_on_c}
	Let $(k,c)$ be a non-trivial $t$-inductive half space and let $k_{\max}, k _{\min}$ denote the entries of $k$ with maximal and minimal absolute value.
	\begin{enumerate}
		\item If $k \geq 0$, we have $c < k_{\max} \cdot k_{\min} + k \cdot \tMinus$.
		\item If $k \leq 0$, we have $c \geq -k_{\max} \cdot k_{\min} + k \cdot \tMinus$.
	\end{enumerate}
\end{restatable}
%---------------------------------------------------------------------
%---------------------------------------------------------------------
%---------------------------------------------------------------------
%---------------------------------------------------------------------

The theorem enforces termination and correctness of \cg.
In fact, we only need to compute sums $k \cdot \tMinus + \sum_{i=1}^\ell k_i$ with $k_1 \dots k_i \in K^*$ up to the limit given in the theorem and still find all values $c$ such that $(k,c)$ is $t$-inductive.
Since the limit is polynomial in the values of $k$ and $t$, \cg~runs in pseudopolynomial time.

We employ \cg~within our CEGAR loop.
Assume we have a solution $k$ to our SMT-formula $\closedm$.
It is left to decide whether there exists a $c \in \N$ such that $(k,c)$ is an IHS.
We apply \cg~to $k$ and each transition $t$.
This yields a set $C_t$ containing all $c_t$ such that $(k,c_t)$ is $t$-inductive and separates $\amark_0$ from $\amark_f$.
Hence, the intersection $\bigcap_{t \in T} C_t$ contains all $c$ such that $(k,c)$ is an IHS that separates $\amark_0$ from $\amark_f$. 
Algorithmically, we only need to test the intersection for non-emptiness.
%---------------------------------------------------------------------
%---------------------------------------------------------------------
%---------------------------------------------------------------------
%---------------------------------------------------------------------
\subsection{Complexity}
We prove that deciding $t$-inductivity for a half space $(k,c)$ and transition $t$ is \coNP-complete.
Membership follows from a non-deterministic variant of \ic.
%Further analysis shows that it is in \textbf{L} for unary input and fixed dimension of $k$, in \textbf{co-NL} for unary input, in \textbf{co-CSL} for binary input, and that it is fixed parameter tractable.
Further analyses show that the problem also lies in \FPT~and in \lang{coCSL}, where \CSL~is the class of languages accepted by context-sensitive grammars.
For unary input, it is in \coNL\ and --- if the dimension of $k$ is fixed --- in \L.
We provide details in the appendix.
The interesting part is \coNP-hardness for which we establish a reduction from the \emph{unbounded subset sum problem}~\cite{Garey:1990:CIG:574848}.
%---------------------------------------------------------------------
%---------------------------------------------------------------------
%---------------------------------------------------------------------
%---------------------------------------------------------------------
\begin{theorem}\label{Theorem:InductivityComplexity}
	Checking $t$-inductivity of a half space $(k,c)$ is \coNP-complete.
\end{theorem}
%---------------------------------------------------------------------
%---------------------------------------------------------------------
%---------------------------------------------------------------------
%---------------------------------------------------------------------

Before we elaborate on the reduction, we introduce the \emph{unbounded subset sum problem} (USSP).
An instance consists of a vector $w \in \N^n$ and an integer $d \in \N$.
The task is to decide whether there exists a vector $x \in \N^n$ such that $w \cdot x = d$.
The problem is \NP-complete~\cite{Garey:1990:CIG:574848}.
To prove \autoref{Theorem:InductivityComplexity}, we reduce from USSP to the complement of checking $t$-inductivity.
This yields the desired \coNP-hardness.
%---------------------------------------------------------------------
%---------------------------------------------------------------------
%---------------------------------------------------------------------
%---------------------------------------------------------------------
\begin{proof}
	Let $(w,d)$ be an instance of USSP.
	We construct a half space $(k,c)$ and a Petri net with a transition $t$ such that $(k,c)$ is not $t$-inductive if and only if there is an $x \in \N^n$ such that $w \cdot x = d$.
	We rely on the inductivity criterion from \autoref{Theorem:inductive}.
	The main difference between this criterion and USSP is that the latter requires reaching a precise value $d$, while the former requires reaching an interval.
	The idea is to define an appropriate half space $(k,c)$ and a transition $t$ such that in the corresponding interval only one value might be reachable.
	
	We set $k = w$. 
	Note that we can assume that $d$ is a multiple of $\gcd(k)$.
	Otherwise, $(w,d)$ is a no-instance of USSP since each linear combination $w \cdot x$ is a multiple of $\gcd(k)$.
	By using the Euclidean algorithm, we can compute an $a \in \Z^n$ such that $k \cdot a = \gcd(k)$ in polynomial time~\cite{Bach1996}.
	%We construct the uniquely defined smallest Petri net with $n$ places and one transition $t$ such that $t^\Delta=-a$.
	We construct a Petri net with $n$ places and one transition $t$ with $\tMinus (i)\define a(i)$ if $a(i)> 0$, $t^+\define -a$ if $a(i)< 0$, and $0$ otherwise. It holds $t^\Delta=-a$.
	Set $c = d + k \cdot \tMinus$.
	It is left to show that $(k,c)$ is not $t$-inductive if and only if $(w,d)$ is a yes-instance of USSP.
	
	Assume that $(k,c)$ is not $t$-inductive.
	Then there exists a vector $x \in \N^n$ such that $c \leq k \cdot x + k \cdot \tMinus < c - k \cdot \tDelta$.
	By plugging in the above definitions, we obtain that $d \leq k \cdot x < d + \gcd(k)$.
	Since $d$, $d + \gcd(k)$, and $k \cdot x$ are all multiples of $\gcd(k)$, we obtain that $d = k \cdot x = w \cdot x$.
	Hence, $(w,d)$ is a yes-instance of USSP.
	
	For the other direction, let $w \cdot x = d$.
	We obtain that $d \leq k \cdot x < d + \gcd(k)$.
	As above, we can employ the definitions and derive that $c \leq k \cdot x + k \cdot \tMinus < c - k \cdot \tDelta$.
	This shows non-inductivity of $(k,c)$ and proves correctness of the reduction.
	\qed
\end{proof}

\section{Experiments}
\label{sec:pn:experiments}
We implemented the CEGAR loop in our Java prototype tool \inequalizer~\cite{inequalizer}. 
It employs \zthree~\cite{DeMoura} as a back-end SMT-solver.
The tool makes use of \emph{incremental solving} as well as \emph{minimization}, a feature of \zthree~that guides the CEGAR loop towards more likely candidates of IHS.
Incremental solving reuses information learned from previous queries to \zthree~and minimization prioritizes solutions with minimal values.
Before \inequalizer~starts the CEGAR loop, it uses an SMT-query to check whether there is a separating IHS $(k,c)$ that is trivial for all transitions.
We use minimization to get half spaces that are non-trivial with respect to fewer transitions.
The reason is that non-trivial half spaces are harder to find and typically only a few values for $c$ ensure inductivity in this case.
%---------------------------------------------------------------------
%---------------------------------------------------------------------
%---------------------------------------------------------------------
%---------------------------------------------------------------------
\begin{table}[t]
	\resizebox{\textwidth}{!}{
		
% !TEX spellcheck = en-EN

\begin{tabular}{|l|c|r|r|r|r|r|c|r|}
\hline
\multicolumn{1}{|c|}{\multirow{2}{*}{Benchmark}} &\multicolumn{1}{c|}{\multirow{2}{*}{$|P|$}}
&\multicolumn{1}{c|}{\multirow{2}{*}{$|T|$}}
& \multicolumn{1}{c|}{\multirow{2}{*}{\inequalizer}}
& \multicolumn{5}{c|}{\mist} \\
\cline{5-9}
  & & & & \multicolumn{1}{c|} {backward} & \multicolumn{1}{c|}{ic4pn}  & \multicolumn{1}{c|}{tsi} & \multicolumn{1}{c|}{eec} & \multicolumn{1}{c|}{eec-cegar}\\
\hline
\bench{BasicME}& 5& 4 & 0.6 & 0.1 & 0.1 & 0.1 & 0.1 & 0.1 \\
\bench{Kanban}&16 &14 & 0.7  & 0.1& 0.2& 0.8 & 0.1 & 0.2 \\
\bench{Lamport}& 11& 9& T/O &0.1 & 0.1 & 0.1 & 0.1 & 0.1 \\
\bench{Manufacturing} &13 &6 & 0.6 & 1.9 & 0.1 & 0.1 & 0.1 & 0.1 \\
\bench{Petersson} & 14& 12 & T/O & 0.2 & 0.1 & 0.1 & 0.1 & 0.1 \\
\bench{Read-write} & 13& 9 & 0.5  &0.1 & 1 & 0.1 & 0.9 & 0.5\\
\bench{Mesh2x2} & 32& 32 & 1.2  & 0.3 & 0.1 & 48.6 & 0.8 & 0.2 \\
\bench{Mesh3x2} & 52& 54 & 2.1  & 2.2& 0.2& T/O & T/O & 2.2 \\
\bench{Multipool} & 18& 21  & 0.8  & 0.3 & 2  & 2.2 & 1 & 2.3 \\
\hline
\end{tabular}
	}
	\caption{\inequalizer\ vs. \mist.}
	\label{tab:mist}
	\vspace*{-0.4cm}
\end{table}
%---------------------------------------------------------------------
%---------------------------------------------------------------------
%---------------------------------------------------------------------
%---------------------------------------------------------------------
Before we show the applicability of \inequalizer~on larger benchmarks, let us consider the Petri net in \autoref{Fig:Runexample}.
When executing \inequalizer, we find that there are no trivial separating IHS.
Using incremental solving, \inequalizer~performs three iterations of the CEGAR loop and returns the non-trivial separating IHS with $k = (53, 52)$ and $c = 209$.
When enabling minimization, we only require two iterations and obtain $k = (8, 5)$, $ c= 22$.
The difference in iterations is due to that we expect minimization to choose vectors $k$ that are trivial for many transitions. This increases the chance of finding a suitable $c$.
On the other hand, incremental solving improves the running time in executions with more iterations.

We evaluated \inequalizer~for \probnamecov~on a benchmark suite and compared it to various methods for coverability implemented in \mist~\cite{mist,GEERAERTS2006180,inproceedings,Ganty2007SymbolicDS,10.1007/3-540-46002-0_13}. 
Results are given in~\autoref{tab:mist}.
The experiments were performed on a 1,7 GHz Intel Core i7 with 8GB memory.
The running times are given in seconds.
For entries marked as T/O, the timeout was reached.
The running times of \inequalizer~are similar to \mist although the former has a small overhead from generating the SMT-query.
In each of the listed Petri nets, the unsafe marking is not coverable.
Except for the mutual exclusion nets \bench{Petersson} and \bench{Lamport}, \inequalizer~reliably finds separating IHS.
Surprisingly, each found IHS is trivial.
We suspect that the cases where \inequalizer~timed out are actually negative instances of \probnamecov.

The experiments show that many practical instances admit trivial IHS, which we synthesize using only one SMT-query.
To test the generation of non-trivial IHS, we ran \inequalizer~on a list of nets that do not  admit trivial ones.
The results are given in the appendix.
They show that \inequalizer~finds non-trivial IHS within few iterations of the CEGAR loop.

\section{Conclusion and Outlook}
%\flo{hier ist vergangenheitsform. davor wares above und below. einheitlich?}
We considered an invariant-based approach to disprove reachability and coverability in Petri nets. %formalized in the problems \probnamereach~and \probnamecov.
The idea was to synthesize an inductive half space that over-approximates the reachable markings of the net and separates them from unsafe markings.
For the synthesis, we established a structure theory of IHS and derived an SMT-formula which linearly approximates the space of IHS.
We provided two algorithms, \ic and \cg.
The former decides whether a half space is inductive, the latter generates suitable constants that guarantee inductivity.
The SMT-formula and the algorithm \cg~were then combined in a CEGAR loop which attempts to synthesize IHS.
We implemented the loop into our tool \inequalizer.
It combines SMT-queries with efficient heuristics and was capable of solving practical instances in our experiments.

We expect that further structural studies of IHS will improve the efficiency of the CEGAR loop.
This may lead to a tighter approximation of the space of IHS or to an improved refinement step eliminating more than multiples.
It is also an intriguing question whether the problems \probnamereach~and \probnamecov~are decidable.
To tackle this, we are currently examining equivalence classes and normal forms of half spaces and their connection to well-quasi orderings.
\paragraph{Acknowledgements}
We thank Roland Meyer for his ideas and contributions that greatly influenced the work on inductive half spaces at hand.
Moreover, we thank Chrisitan Eder for sharing with us his experience in commutative algebra, and Marvin Triebel for his ideas and the questions that he raised during his visit.
\bibliographystyle{plain}
\bibliography{masterbib}

\begin{thebibliography}{10}

\bibitem{Abdulla2004}
P.~A. Abdulla, S.~P. Iyer, and A.~Nyl{\'e}n.
\newblock Sat-solving the coverability problem for {P}etri nets.
\newblock {\em Formal Methods in System Design}, 24(1):25--43, 2004.

\bibitem{Bach1996}
E.~Bach and J.~Shallit.
\newblock {\em Algorithmic Number Theory, Volume I: Efficient Algorithms}.
\newblock MIT Press, 1996.

\bibitem{10.1007/978-3-540-69738-1_27}
D.~Beyer, T.~A. Henzinger, R.~Majumdar, and A.~Rybalchenko.
\newblock Invariant synthesis for combined theories.
\newblock In {\em {VMCAI}}, volume 4349 of {\em LNCS}, pages 378--394.
  Springer, 2007.

\bibitem{Blanchet:2003:SAL:781131.781153}
B.~Blanchet, P.~Cousot, R.~Cousot, J.~Feret, L.~Mauborgne, A.~Min{\'e},
  D.~Monniaux, and X.~Rival.
\newblock A static analyzer for large safety-critical software.
\newblock In {\em {PLDI}}, pages 196--207. ACM, 2003.

\bibitem{10.2307/2371684}
A.~Brauer.
\newblock On a problem of partitions.
\newblock {\em American Journal of Mathematics}, 64(1):299--312, 1942.

\bibitem{Cardoza:1976:ESC:800113.803630}
E.~Cardoza, R.~Lipton, and A.~R. Meyer.
\newblock Exponential space complete problems for {P}etri nets and commutative
  semigroups (preliminary report).
\newblock In {\em {STOC}}, pages 50--54. ACM, 1976.

\bibitem{Clarke2018}
E.~M. Clarke, T.~A. Henzinger, H.~Veith, and R.~Bloem.
\newblock {\em Handbook of Model Checking}.
\newblock Springer, 2018.

\bibitem{CEGAR00}
E.M. Clarke, O.~Grumberg, S.~Jha, Y.~Lu, and H.~Veith.
\newblock Counterexample-guided abstraction refinement.
\newblock In {\em CAV}, volume 1855 of {\em LNCS}, pages 154--169. Springer,
  2000.

\bibitem{Cousot:1978:ADL:512760.512770}
P.~Cousot and N.~Halbwachs.
\newblock Automatic discovery of linear restraints among variables of a
  program.
\newblock In {\em {POPL}}, pages 84--96. ACM, 1978.

\bibitem{Czerwinski2019}
W.~Czerwinski, S.~Lasota, R.~Lazic, J.~Leroux, and F.~Mazowiecki.
\newblock The reachability problem for {P}etri nets is not elementary.
\newblock In {\em {STOC}}, pages 24--33. {ACM}, 2019.

\bibitem{Czerwinski2021}
W.~Czerwiński and L.~Orlikowski.
\newblock Reachability in vector addition systems is ackermann-complete, 2021.

\bibitem{10.1007/3-540-46002-0_13}
G.~Delzanno, J.~Raskin, and L.~Van Begin.
\newblock Towards the automated verification of multithreaded java programs.
\newblock In {\em {TACAS}}, pages 173--187. Springer, 2002.

\bibitem{10.1007/978-3-319-08867-9_40}
J.~Esparza, R.~Ledesma-Garza, R.~Majumdar, P.~Meyer, and F.~Niksic.
\newblock An {SMT}-based approach to coverability analysis.
\newblock In {\em {CAV}}, pages 603--619. Springer, 2014.

\bibitem{Esparza99anunfolding}
J.~Esparza and S.~R{\"o}mer.
\newblock An unfolding algorithm for synchronous products of transition
  systems.
\newblock In {\em {CONCUR}}, pages 2--20. Springer, 1999.

\bibitem{Esparza:1996:IMU:646480.693793}
J.~Esparza, S.~R\"{o}mer, and W.~Vogler.
\newblock An improvement of {M}cmillan's unfolding algorithm.
\newblock In {\em {TACAS}}, volume 1055 of {\em LNCS}, pages 87--106. Springer,
  1996.

\bibitem{inequalizer}
{F. Furbach}.
\newblock Inequalizer - a prototype tool for linear safety verification of
  {P}etri nets.
\newblock \url{https://github.com/florianfurbach/Inequalizer}.

\bibitem{Floyd1967Flowcharts}
R.~W. Floyd.
\newblock Assigning meanings to programs.
\newblock {\em Proceedings of a symposium on Applied Mathematics}, 19:19--32,
  1967.

\bibitem{Ganty2007SymbolicDS}
P.~Ganty, C.~Meuter, L.~V. Begin, G.~Kalyon, J.~Raskin, and G.~Delzanno.
\newblock Symbolic data structure for sets of k-uples of integers.
\newblock 2007.

\bibitem{inproceedings}
P.~Ganty, J.~Raskin, and L.~Begin.
\newblock From many places to few: Automatic abstraction refinement for {P}etri
  nets.
\newblock volume~88, pages 124--143, 06 2007.

\bibitem{zbMATH05506898}
P.~Ganty, J.-F. Raskin, and L.~Van Begin.
\newblock From many places to few: automatic abstraction refinement for {P}etri
  nets.
\newblock {\em Fundam. Inform.}, 88(3):275--305, 2008.

\bibitem{Garey:1990:CIG:574848}
M.~R. Garey and D.~S. Johnson.
\newblock {\em Computers and Intractability; A Guide to the Theory of
  NP-Completeness}.
\newblock W. H. Freeman \& Co., 1990.

\bibitem{10.1007/978-3-540-75596-8_9}
G.~Geeraerts, J.~Raskin, and L.~Van Begin.
\newblock On the efficient computation of the minimal coverability set for
  {P}etri nets.
\newblock In {\em ATVA}, pages 98--113. Springer, 2007.

\bibitem{GEERAERTS2006180}
G.~Geeraerts, J.-F. Raskin, and L.~Van Begin.
\newblock Expand, enlarge and check: New algorithms for the coverability
  problem of {WSTS}.
\newblock {\em Journal of Computer and System Sciences}, 72:180 -- 203, 2006.

\bibitem{Greuel2002}
G.-M. Greuel and G.~Pfister.
\newblock {\em A Singular Introduction to Commutative Algebra}.
\newblock Springer, 2002.

\bibitem{Gulwani:2008:PAC:1375581.1375616}
S.~Gulwani, S.~Srivastava, and R.~Venkatesan.
\newblock Program analysis as constraint solving.
\newblock In {\em {PLDI}}, pages 281--292. ACM, 2008.

\bibitem{gupta2009tests}
A.~Gupta, R.~Majumdar, and A.~Rybalchenko.
\newblock From tests to proofs.
\newblock In {\em {TACAS}}, pages 262--276. Springer, 2009.

\bibitem{Hartmanis1967}
J.~Hartmanis.
\newblock Context-free languages and turing machine computations.
\newblock In {\em Symposia in Applied Mathematics}, volume~19, pages 42--51,
  1967.

\bibitem{Hoare:1969:ABC:363235.363259}
C.~A.~R. Hoare.
\newblock An axiomatic basis for computer programming.
\newblock {\em Commun. ACM}, 12(10):576--580, 1969.

\bibitem{10.1007/978-3-642-32940-1_35}
A.~Kaiser, D.~Kroening, and T.~Wahl.
\newblock Efficient coverability analysis by proof minimization.
\newblock In {\em {CONCUR}}, pages 500--515. Springer, 2012.

\bibitem{KARP1969147}
R.~M. Karp and R.~E. Miller.
\newblock Parallel program schemata.
\newblock {\em Journal of Computer and System Sciences}, 3(2):147 -- 195, 1969.

\bibitem{10.1007/978-3-642-39799-8_10}
J.~Kloos, R.~Majumdar, F.~Niksic, and R.~Piskac.
\newblock Incremental, inductive coverability.
\newblock In {\em {CAV}}, pages 158--173. Springer, 2013.

\bibitem{Kosaraju1982}
S.~Rao Kosaraju.
\newblock Decidability of reachability in vector addition systems.
\newblock In {\em {STOC}}, pages 267--281. {ACM}, 1982.

\bibitem{Lambert1992}
J.~Lambert.
\newblock A structure to decide reachability in {P}etri nets.
\newblock {\em Theor. Comput. Sci.}, 99(1):79--104, 1992.

\bibitem{Langerak:1999:CFP:647768.733925}
R.~Langerak and E.~Brinksma.
\newblock A complete finite prefix for process algebra.
\newblock In {\em {CAV}}, volume 1633 of {\em LNCS}, pages 184--195. Springer,
  1999.

\bibitem{Leroux2021}
J.~Leroux.
\newblock The reachability problem for petri nets is not primitive recursive,
  2021.

\bibitem{Leroux2019}
J.~Leroux and S.~Schmitz.
\newblock Reachability in vector addition systems is primitive-recursive in
  fixed dimension.
\newblock In {\em {LICS}}, pages 1--13. {IEEE}, 2019.

\bibitem{yale1976reachability}
R.~J. Lipton.
\newblock {\em The reachability problem requires exponential space}.
\newblock Research report (Yale University. Department of Computer Science).
  Department of Computer Science, Yale University, 1976.

\bibitem{Mayr1981}
E.~Mayr.
\newblock An algorithm for the general {P}etri net reachability problem.
\newblock In {\em {STOC}}, pages 238--246. {ACM}, 1981.

\bibitem{grr}
E.~Mayr.
\newblock An algorithm for the general {P}etri net reachability problem.
\newblock {\em SIAM Journal on Computing}, 13(3):441--460, 1984.

\bibitem{McMillan:1995:TSS:204526.204528}
K.~L. McMillan.
\newblock A technique of state space search based on unfolding.
\newblock {\em Form. Methods Syst. Des.}, 6(1):45--65, 1995.

\bibitem{DeMoura}
L.~De Moura and N.~Bj{\o}rner.
\newblock Z3: An efficient {SMT} solver.
\newblock In {\em {TACAS}}, volume 4963 of {\em {LNCS}}, pages 337--340.
  Springer, 2008.

\bibitem{24143}
T.~Murata.
\newblock {P}etri nets: Properties, analysis and applications.
\newblock {\em Proceedings of the IEEE}, 77(4):541--580, 1989.

\bibitem{mist}
{P. Ganty}.
\newblock mist - a safety checker for {P}etri nets and extensions.
\newblock \url{https://github.com/pierreganty/mist}.

\bibitem{Peterson:1981:PNT:539513}
J.~L. Peterson.
\newblock {\em {P}etri Net Theory and the Modeling of Systems}.
\newblock Prentice Hall PTR, 1981.

\bibitem{10.1007/978-3-540-24622-0_20}
A.~Podelski and A.~Rybalchenko.
\newblock A complete method for the synthesis of linear ranking functions.
\newblock In {\em {VMCAI}}, pages 239--251. Springer, 2004.

\bibitem{RACKOFF1978223}
C.~Rackoff.
\newblock The covering and boundedness problems for vector addition systems.
\newblock {\em Theoretical Computer Science}, 6(2):223 -- 231, 1978.

\bibitem{10.1007/978-3-642-21834-7_5}
P.~Reynier and F.~Servais.
\newblock Minimal coverability set for {P}etri nets: {K}arp and {M}iller
  algorithm with pruning.
\newblock In {\em {PETRI NETS}}, pages 69--88. Springer, 2011.

\bibitem{Sankaranarayanan2003}
S.~Sankaranarayanan, H.~Sipma, and Z.~Manna.
\newblock {\em {P}etri Net Analysis Using Invariant Generation}, pages
  682--701.
\newblock Springer, 2003.

\bibitem{10.1007/978-3-540-30579-8_2}
S.~Sankaranarayanan, H.~Sipma, and Z.~Manna.
\newblock Scalable analysis of linear systems using mathematical programming.
\newblock In {\em {VMCAI}}, volume 3385 of {\em LNCS}, pages 25--41. Springer,
  2005.

\bibitem{Sipser1997}
M.~Sipser.
\newblock {\em Introduction to the theory of computation}.
\newblock {PWS} Publishing Company, 1997.

\bibitem{Triebel2015}
M.~Triebel and J.~S{\"{u}}rmeli.
\newblock Characterizing stable inequalities of {P}etri nets.
\newblock In {\em {PETRI NETS}}, volume 9115 of {\em LNCS}, pages 266--286.
  Springer, 2015.

\bibitem{Turing1937}
A.~M. Turing.
\newblock On computable numbers, with an application to the
  entscheidungsproblem.
\newblock {\em Proceedings of the London Mathematical Society},
  s2-42(1):230--265, 01 1937.

\bibitem{10.1007/978-3-642-31131-4_12}
A.~Valmaris and H.~Hansen.
\newblock Old and new algorithms for minimal coverability sets.
\newblock In {\em {PETRI NETS}}, pages 208--227. Springer, 2012.

\end{thebibliography}
%
%%-----------------------------------------------------------------------------------------------------
%%-----------------------------------------------------------------------------------------------------
%%-----------------------------------------------------------------------------------------------------
%
\appendix
\appendixpage
\addappheadtotoc

\section{Proof of \autoref{Theorem:inductive}}\label{sec:pn:app:firstproof}
\inductive*
\begin{proof}
	We use the following property:
	%The minimum of a set is less or equal some value $z$ iff the set contains an element less or equal $z$.
	%\begin{equation}\label{eq:infy}
	% inf_{t,s} (k\cdot P \geq c)\leq z \Leftrightarrow \exists y\in M_t \cap M(k\cdot P \geq c ) : k\cdot y\leq z
	%\end{equation}
	\begin{equation}\label{eq:yx}
		\exists \amark\in \mathbb{N}^n: \amark\geq t^- \Leftrightarrow \exists x\in \mathbb{N}^S: \amark=x+t^- .
	\end{equation}
	According to the definition of t-inductivity, a half space is not t-inductive iff there is a marking $\amark$ that satisfies the following condition:
	\begin{align*}
		\exists \amark\in \mathbb{N}^n:&\; k\cdot \amark \geq c \wedge  \amark\geq t^- \wedge  k\cdot \amark +k\cdot t^\Delta < c\\
		\stackrel{(\ref{eq:yx})}{\Leftrightarrow} \exists x\in \mathbb{N}^n:&\; k\cdot x +k\cdot t^- \geq c  \wedge  k\cdot x +k\cdot t^- +k\cdot t^\Delta < c \\
		\Leftrightarrow \exists x\in \mathbb{N}^n:&\; k\cdot x +k\cdot t^- \geq c  \wedge  k\cdot x +k\cdot t^-  < c -k\cdot t^\Delta\\
		\exists x\in \mathbb{N}^n:&\; c \leq k\cdot x + k \cdot t^- < c- k\cdot t^\Delta
	\end{align*}
\end{proof}
%-------------------------------------------------------------
%-------------------------------------------------------------
%-------------------------------------------------------------
%-------------------------------------------------------------
%-------------------------------------------------------------
%-------------------------------------------------------------
%-------------------------------------------------------------
%-------------------------------------------------------------
\section{Proofs for \autoref{Section:Charinduct}}
\begin{lemma}
	Let $(k, c)$ be oriented towards $t$. 
	If $\amark\in \Act{t} \cap \Sol{k}{c}$, then $\amark+\tDelta\in \Sol{k}{c}$.
\end{lemma}
\begin{proof}
	The lemma holds by $k\cdot(\amark + \tDelta) = k\cdot \amark + k\cdot \tDelta \geq k\cdot \amark\geq c$. 
	The first equality is by distributivity of the scalar product, the following inequality is by the definition of orientation towards $t$, and the last inequality is by $\amark\in\Sol{k}{c}$. 
	Finally, $(\amark + \tDelta) \in \Sol{k}{c}$ follows from $k\cdot(\amark + \tDelta)\geq c$ according to the definition of $\Sol{k}{c}$.
\end{proof}
%-------------------------------------------------------------
%-------------------------------------------------------------
%-------------------------------------------------------------
%-------------------------------------------------------------

\begin{lemma}
	Let $(k, c)$ be monotone for $t$. If $\amark\in \Act{t} \cap \Sol{k}{c}$, then $\amark+\tDelta\in \Sol{k}{c}$.
\end{lemma}
\begin{proof}
	Note that membership $\amark\in\Act{t}$ implies $\amark = \tMinus+\avec$ with $\avec\in \N^n$. 
	We have $k\cdot (\amark + \tDelta) = k\cdot (\tMinus + \tDelta + \avec) = k\cdot (\tMinus+\tDelta) + k\cdot \avec \geq c$.
	The first equality rearranges the terms of $\amark+\tDelta$, the second is distributivity, the third is by the fact that $k\cdot (\tMinus+\tDelta)\geq c$ and $k\cdot \avec\geq 0$.
	The former holds by the definition of monotonicity for half spaces, the latter is by the fact that $k, c, \avec \geq 0$. It holds $\amark+\tDelta\in \Sol{k}{c}$ since $k\cdot (\amark + \tDelta)\geq c$.
\end{proof}
%-------------------------------------------------------------
%-------------------------------------------------------------
%-------------------------------------------------------------
%-------------------------------------------------------------
\begin{lemma}
	If $(k, c)$ is antitone for $t$, then $\Act{t} \cap \Sol{k}{c}=\emptyset$. 
\end{lemma}

\begin{proof}
	Note that membership $\amark\in \Act{t}$ means $\amark = \tMinus+\avec$ for some $\avec\in \N^n$. 
	Now $k\cdot \amark = k\cdot (\tMinus+\avec) = k\cdot \tMinus+k\cdot \avec < c$ follows. The inequality is by $k\cdot \tMinus<c$ and $k\cdot \avec\leq 0$. 
	The former holds by the definition of antitone half spaces, the latter by $k\leq 0$ and $\avec\geq 0$. 
\end{proof} 

\NecessNontrivial*
\begin{proof}
	Assume towards contradiction that there is a $k(i)$ with $0<|k(i) |\leq -k\mal \trans$
	Since it is non-trivial, we can apply \autoref{Lemma:Nontrivial} and either (a) or (b) holds. 
	%$k$ is not mixed, so either $k\leq 0$ or $k\geq 0$ holds. 
	%Since neither Def.\ref{def:trivial}.2) nor Def.\ref{def:trivial}.3) holds, we need to examine two cases,
 %either $k\mal t^- <c-k\mal \trans $ and $k\geq 0$ or $k\mal t^- \geq c$ and $k\leq 0$ hold. 
	For both cases, we define a value $z$ such that $c\leq k\mal t^- +z\mal k(i)<c-k\mal \trans$ is satisfied.
	
	\paragraph*{Case (a) $k\mal t^- <c-k\mal \trans $ and $k\geq 0$:} 
	There is a smallest $z\in \N$ such that $c\leq k\mal t^- +z\mal k(i)$. 
	If $k\mal t^-\geq c$, then $z=0$ holds trivially. Note that  $k\mal t^- +z\mal k(i)=k\mal \tMinus <c-k\mal \trans$ holds by (a).
	If $k\mal t^- < c$, then $z$ exists since $k(i)>0$.
	Assume towards contradiction that $c- k\mal \tDelta \leq k\mal t^- +z\mal k(i)$. It holds 
	$$ k\mal t^- +(z-1)\mal k(i)\geq c-k\mal \tDelta- k(i)\geq c.$$
	The first inequality is by the assumption  $c- k\mal \tDelta \leq k\mal t^- +z\mal k(i)$ and the second is by $|k(i) |\leq -k\mal \trans$.
	This means $z-1$ also satisfies the condition which is a contradiction to $z$ being smallest. 	
	The condition $c\leq k\mal t^- +z\mal k(i)<c-k\mal \trans$ is satisfied either way.
			
	\paragraph*{Case (b) $k\mal t^- \geq c$ and $k\leq 0$:} 
	There is a smallest $z\in \N$ such that $k\mal t^- +z\mal k(i)<c-k\mal \trans$. 
	Since $z$ is smallest and $|k(i) |\leq -k\mal \trans$ holds it follows
	$c\leq k\mal t^- +z\mal k(i)<c-k\mal \trans$ as well. The detailed proof is omitted since it is analogue to (a).
	\smallskip
	
	We define the vector $x\in \N^n$ as $x=z\mal e_i$. 
	According to \autoref{Theorem:inductive}, this implies that $(k,c)$ is not t-inductive which is a contradiction.
	\qed
\end{proof}

%---------------------------------------------------------------------
%---------------------------------------------------------------------
%---------------------------------------------------------------------
%---------------------------------------------------------------------
%---------------------------------------------------------------------
%---------------------------------------------------------------------
%---------------------------------------------------------------------
%---------------------------------------------------------------------
\section{Proofs for \autoref{Section:Generating}}
\label{sec:pn:app:proofgenc}
In order to prove \autoref{lem:genc}, we first require the following lemma which examines the non-trivial case:
\begin{lemma}\label{lemma_turning_k(i)nto_invariant}
	Let $k$ be an non-mixed vector with $\ktdelta < 0$ and $\abs{k(i)} > -\ktdelta$ for all $k(i) \neq 0$. Then there is a $c \in \Z$ so that \ineq is $t$-inductive.
\end{lemma}

\begin{proof}
	We examine both cases for the non-mixed vector $k$: $k\geq 0$ and $k\leq 0$.
	\begin{description}
		\item[$k \geq 0$:] 		
		First, we assume $k \geq 0$ and set $c \define k\mal t^{-} + 1$. From $\ktdelta < 0$ and $\abs{k(i)} > -\ktdelta$ follows $k(i) >1$ for all $k(i)\neq 0$. 
		We check if the condition in \autoref{Theorem:inductive} is satisfied for any $x\in \N^n$:
		$$c \leq k\cdot x + k \cdot t^- < c- k\cdot t^\Delta.$$
		If $x$ is the vector $0^n$ then it follows $c=k\mal t^- +1 \nleq k\mal x +k\mal t^- =k\mal t^-$.\newline
		Let $x$ contain an entry $x(i)>0$. If $k(i)=0$ for all such entries then the case is analogue to $x=0^n$. 
		If there is an $i\leq n$ with $x(i)>0$ and $k(i)>0$ (and thus $k(i)\geq -\ktdelta +1$ ), then it holds 
		$$k\mal x+k\mal t^-\geq k(i)+k\mal t^-\geq  -\ktdelta+1+k\mal t^-=c-\ktdelta.$$
		\item[$k \leq 0$:] 	
		We define $c \define k \mal t^{-} + \ktdelta$. 
		If $x$ is the vector $0^n$, then it follows $k\mal x +k\mal t^- =k\mal t^-\nless  k t^{-} = c - \ktdelta$.
		Let $x$ contain an entry $x(i)>0$. If $k(i)=0$ for all such entries then the case is analogue to $x=0^n$. 
		If there is an $i\leq n$ with $x(i)>0$ and $k(i)<0$ (and thus $k(i)<\ktdelta$ ), then it holds 
		$$k\mal x+k\mal t^- \leq k(i)+k\mal t^-<  \ktdelta+k\mal t^-=c.$$
	\end{description}
	This means the condition is not satisfied in either case and the half space is t-inductive according to \autoref{Theorem:inductive}.
\end{proof}
We recall the \autoref{Theorem:inductive} from \autoref{Section:Generating}:
\genc*
\begin{proof}
Now, we examine the formula \closedform.
Let $(k,c)$ be separating, meaning $\amark_0 \in \Sol{k}{c}$ and $\amark_0 \notin \Sol{k}{c}$, and $t$-inductive. Since it is separating, condition (0) holds and $c$ is in the interval 
$(k\cdot \amark_f,k\cdot \amark_0 ]$.
If it is trivial, it follows that one of the conditions (\ref{eq:gen1})-(\ref{eq:gen3}) holds.
If it is non-trivial, conditions (4) and (5) hold. It follows that $k$ is a satisfying assignment of $\closedmt$.

Let $k$ be a solution of $\closedmt$. We show that there is a value $c$ such that $(k,c)$ is t-inductive:
\begin{itemize}
	\item If $k$ satisfies (\ref{eq:gen1}), it is $t$-inductive for any $c$ and since (0) holds, we can choose a $c$ such that it is separating.
	\item If $k$ satisfies (\ref{eq:gen2}) then we set $c=k\cdot \amark_0$ and thus $(k,c)$ is separating and antitone.
	\item Condition (\ref{eq:gen3}) is analogue, here we set $c=k\cdot \amark_f +1$ and thus $(k,c)$ is separating and monotone.
	\item If $k$ satisfies Conditions (4) and (5), than the property holds according to the following Lemma.
\end{itemize}
\end{proof}
%---------------------------------------------------------------------
%---------------------------------------------------------------------
%---------------------------------------------------------------------
%---------------------------------------------------------------------
\lemmamult*
\begin{proof}
	We use contraposition. 
	Given a marking $\amark$ that violates t-inductivity of $(k, \lceil \frac c a \rceil )$. 
	We show that it violates t-inductivity of $(k',c)$ as well.
	
	If it holds $k\mal \amark \geq \lceil \frac c a \rceil$, then it follows $ k\mal \amark \geq \frac c a$ from $\lceil \frac c a \rceil\geq \frac c a$ and thus $ k'\mal \amark \geq c$.
	 %Note that $\frac{c}{a}$ may not be an integer but $k\mal \amark$ is an integer that is at least  $\frac{c}{a}$ and thus at least $\lceil \frac c a \rceil$. This means  $(k, \lceil \frac c a \rceil )$ holds. 
	If it holds  $k\mal (\amark +\tDelta)<\lceil \frac c a \rceil $ then it follows 
	 $k\mal (\amark +\tDelta)+1 \leq \lceil \frac c a \rceil $. Finally, we conclude  
	 $k'\mal (\amark +\tDelta)<c$ using $\lceil \frac c a \rceil \leq \frac c a +1$.\qed
\end{proof}

\begin{theorem}
		Let $(k,c)$ be such that $\amark_f\notin \Sol{k}{c}$. It holds $\amark_f \uparrow \cap \Sol{k}{c} =\emptyset$ iff $k\leq 0$.
\end{theorem}
\begin{proof}
	"$\Rightarrow$:" Let $m_f \uparrow \cap \Sol{k}{c} =\emptyset$. 
	We assume towards contradiction that $k(i)>0$ holds for some $i\leq n$. Let $m$ be such that $m(i)=m_f(i)+(c-k\mal m_f)$ and $m(j)=m_f(j)$ for all $j\leq n$ with $j\neq i$.
	Then $k\mal m= k\mal m_f +k(i)\mal (c-k\mal m_f)$. Since $k(i)>0$ and $(c-k\mal m_f)>0$, it follows $k\mal m\geq k\mal m_f +1\mal (c-k\mal m_f)=k\mal m_f +c-k\mal m_f=c\geq c$ and thus $m\in m_f \uparrow \cap \Sol{k}{c} \neq \emptyset$. This is a contradiction to $m_f \uparrow \cap \Sol{k}{c} =\emptyset$.
	
	\noindent "$\Leftarrow$:" Let $k\leq 0$ and $m\in m_f \uparrow$ and thus $m\geq m_f$.
	It follows $k\mal m \leq k \mal m_f<c$ and thus $m\notin \Sol{k}{c}$. This means $m_f \uparrow \cap \Sol{k}{c} =\emptyset$.
\end{proof}
\section{Proofs for \autoref{Sec:checking}}\label{sec:pn:app:algo}
We recall \autoref{Theorem:thm1} as well as  \autoref{alg:Infinum}. 
	\begin{restatable}{theorem}{thmone}
			\label{Theorem:thm1}
			A half space $(k,c)$ is t-inductive iff
$$\nexists \sumk{1} \ldots \sumk{l} \in K^* : c \leq k \cdot t^- + \sum_{i=1}^{l} \sumk{i} < c- k\cdot t^\Delta$$
	\end{restatable}
%\thmone*
		\begin{proof}
			We use the following property:
			%The minimum of a set is less or equal some value $z$ iff the set contains an element less or equal $z$.
			%\begin{equation}\label{eq:infy}
			% inf_{t,s} (k\cdot P \geq c)\leq z \Leftrightarrow \exists y\in M_t \cap M(k\cdot P \geq c ) : k\cdot y\leq z
			%\end{equation}
			For any sum $\sum_{i=1}^l \sumk{i}$ over $K$ with $x(j)$ the number of occurrences of value $k_j$ in the sum, it holds
			\begin{equation}\label{eq:sumprop}
			\sum_{i=1}^l \sumk{i}=\sum_{j=1}^{n} k(j)\cdot x(j)=k\cdot x.
			\end{equation}
			Obviously, we can construct a sum with value $k\cdot x$ for any vector $x$. It follows:
			\begin{equation}\label{eq:sumy}
			\exists x\in \mathbb{N}^n: k\cdot x= z \Leftrightarrow \exists \sumk{1}\ldots \sumk{l} \in K^*: \sum_{i=1}^l \sumk{i}=z
			\end{equation}
			%$$inf_{t,s} (k\cdot P \geq c)+k\cdot t^\Delta < c $$
			%$$\stackrel{(\ref{eq:infy})}{\Leftrightarrow}  
			A half space is not t-inductive iff there is a marking $x$ that violates \autoref{Theorem:inductive}:
			\begin{align*}
				\exists x\in \mathbb{N}^n:&\; c \leq k\cdot x + k \cdot t^- < c- k\cdot t^\Delta\\
				\stackrel{(\ref{eq:sumy})}{\Leftrightarrow} \exists \sumk{1} \ldots \sumk{l} \in K^*:&\; c \leq \sum_{i=1}^{l} \sumk{i} + k \cdot t^- < c- k\cdot t^\Delta
			\end{align*}
		\end{proof}
	\begin{lemma}\label{lem:pn:indcorrect}
		The algorithm is correct.
	\end{lemma}
	\begin{proof}
		Assume the algorithm returns "Not inductive".
		It holds $c\leq current <c-k\mal t^\Delta $ and $current$ was derived by adding elements of $K$ to the starting value $k\mal t^-$.
		It follows that there is a sequence $\sumk{1} \ldots \sumk{l} \in K^*$ and $c\leq k \mal t^- +\sum_{i=1}^l \sumk{i} <c-k\mal t^\Delta $.
		This violates the condition of \autoref{Theorem:thm1} and thus the half space is not t-inductive.
	\end{proof}
	%\begin{lemma}
	% $$\exists k_1\ldots k_l \in K^*:\sum_{i=1}^{l}=z \Leftrightarrow$$
	% $$\exists k'_1 \ldots k'_l \in K^*:\sum_{i=1}^{l}=z \wedge \forall i<n: \left( \sum_{j=1}^i k'_j<z \leftrightarrow k'_{i+1}\in K^*\right)$$
	%\end{lemma}
	%\begin{proof}
	%\end{proof}
	
	\begin{lemma}\label{lem:pn:indcomplete}
		The algorithm is complete.
	\end{lemma}
	\begin{proof}
		We show that the algorithm identifies any half space that is not t-inductive.
		We know that $k$ is not mixed. We now examine the case $k\geq 0$. %According to Lemma~\ref{lem:3con} it holds $k\mal t^-<c$
		Let $k_1 \ldots k_l$ be a sequence that satisfies the condition of \autoref{Theorem:thm1}. 
		%W.l.o.g. we can assume that there is no subsequence of $k_1 \ldots k_l$  that sums up to $0$ since we could simply remove it without changing the sum. 
		W.l.o.g., we can assume that $k_1 \ldots k_l$ is the shortest sequence that reaches the target area. 
		We examine its prefixes $k_1 \ldots k_i$ with $i<l$.
		It follows that $k\mal t^-+\sum_{j=1}^i \sumk{j}< c$ holds for all $i<l$ and $ k\mal t^-+\sum_{j=1}^{l} \sumk{j}<c-k\mal \trans$.
		%Let $x\in \mathbb{N}^n$ be such that $ c \leq k\mal x + k \mal t^- < c- k\mal t^\Delta$. W.l.o.g. we can assume that there is no $0<y\leq x$ such that $k\mal y=0$ since $k\mal (x-y)=k\mal x$ and thus $x-y$ is a solution as well. Furthermore, we can assume that there is no $0<y\leq x$ such that $ c \leq k\mal y + k \mal t^- < c- k\mal t^\Delta$.
		%There is a sequence $k_1 \ldots .k_l$ such that $|\{k_i \mid k_i=x(j)\}|=x_j$ and $l=|x|$. It follows $k\mal x= \sum_{i=1}^l k_i$.
		%For any $i<l$ it holds that if $k\mal t^-+\sum_{j=1}^i k_j< c$, then $\sum_{j=1}^i k_j<\sum_{j=1}^l k_j$ and thus there is a $k_j\in K^+$ with $j>i$.
		%For any $i<l$ it holds that if $k\mal t^-+\sum_{j=1}^i k_j\geq c-k\mal t^\Delta$, then $\sum_{j=1}^i k_j>\sum_{j=1}^l k_j$ and thus there is a $k_j\in K^-$ with $j>i$. Since the sequence is shortest, the sum is outside the target area for every $j<l$ which means they all satisfy one of the two conditions.
		%Since addition is commutative, we can rearrange the sequence and move those $k_j$ to the front of the suffix.
		%It follows that there is a sequence $k'_1 \ldots k'_l$ that additionally satisfies:
		%\begin{equation}\label{eq:k+}
		%	\forall_{i<l}: k\mal t^-+\sum_{j=1}^i k'_j< c \Rightarrow k'_{i+1}\in K^+
		%\end{equation}
		%\begin{equation}\label{eq:k-}
		%	\forall_{i<l}: k\mal t^-+\sum_{j=1}^i k'_j\geq c-k\mal t^\Delta \Rightarrow k'_{i+1}\in K^-
		%\end{equation}	
		
		We now apply a induction over the sequence to show that the algorithm processes $k\mal t^- +\sum_{j=1}^l k_j $ or returns "Not inductive" before that:
		\begin{description}
			\item[Induction basis:] The algorithm processes $k\mal t^-$ (Line~1 and 2). 
			\item[Induction hypothesis:] The algorithm processes $current=k\mal t^- +\sum_{j=1}^i \sumk{j}$.
			\item[Induction step:] 		We know $current+\sumk{i+1}= k\mal t^-+\sum_{j=1}^{i+1} \sumk{j}<c-k\mal \trans$ holds and thus the condition in \autoref{alg:ifadda} is satisfied.
			It is added to the queue (\autoref{alg:queueadd}) and it is either processed later or the algorithm returns "Not inductive" before that.
			The argument is analogue for $k\leq 0$. Here, the condition in \autoref{alg:ifaddb} is satisfied.
		\end{description}
%		\paragraph*{Induction basis:} The algorithm processes $k\mal t^-$ (Line~1 and 2). 
%		\paragraph*{Induction hypothesis:} The algorithm processes $current=k\mal t^- +\sum_{j=1}^i \sumk{j}$.
%		\paragraph*{Induction step:} %Either the condition  If $k\mal t^-+\sum_{j=1}^i k'_j< c$ holds
%		%The condition in \autoref{alg:setk+} is met and the algorithm assigns $Set:=K^+$ when $current$ is processed. 
%		We know $current+\sumk{i+1}= k\mal t^-+\sum_{j=1}^{i+1} \sumk{j}<c-k\mal \trans$ holds and thus the condition in \autoref{alg:ifadda} is satisfied.
%		It is added to the queue (\autoref{alg:queueadd}) and it is either processed later or the algorithm returns "Not inductive" before that.
%		The argument is analogue for $k\leq 0$. Here, the condition in \autoref{alg:ifaddb} is satisfied.
%				
		It follows that unless "Not inductive" is returned earlier, $k\mal t^-+\sum_{j=1}^l \sumk{j}$ is processed. 
		It meets the condition in \autoref{alg:ifuninductive} and the algorithm returns "Not inductive".
	\end{proof}

	\begin{theorem}\label{thm:pol}
	The run-time of \autoref{alg:Infinum} is polynomial in the input values.
	\end{theorem}
%	\begin{lemma}\label{lem:pol}
%	\end{lemma}
	\begin{proof}
	We show that the algorithm's runtime is polynomial in the input values.
		The processing time of one value is linear in $|K|$.
		Either the lowest processed value is  the starting value $k\mal t^-$ and the highest is some value at most $ c-k\mal t^\Delta$ (garanteed by the condition in \autoref{alg:ifadda}) or 
		the highest processed value is the starting value $k\mal t^-$ and the lowest is at least $c$ (garanteed by the condition in \autoref{alg:ifaddb}).
		It follows that the algorithm only processes values in the polynomial sized segment 
		$$[min(k\mal t^-,c) , max(k\mal t^-,c-k\mal t^\Delta)].$$
		Since every processing step reaches a new unprocessed value in the segment, the number of processing steps are limited by the seqment size
		$$l_s:= | max(k\mal t^-,c-k\mal t^\Delta)-min(k\mal t^-,c) |.$$
	\end{proof}
	
	We continue by analyzing the space-complexity of the problem. We study the complexity class \textbf{L} which denotes problems that can be solved deterministically using an amount of memory space that is logarithmic in the size of the input. The class \textbf{NL} describes problems that can be solved non-deterministically in logarithmic space. The problems that can be solved non-deterministically using space that is linear in the input size are in \textbf{CSL}. For any class of problems \textbf{C}, we denote the class of their complements as \textbf{co-C}.
	%We will show that it is in \textbf{NL} and we hope to prove either \textbf{NL}-hardness or membership in \textbf{Logspace}.

	\begin{theorem}\label{Theorem:smallx}
		A half space is not t-inductive if and only if there is a vector $x$ that violates \autoref{Theorem:inductive} and $x(j)\leq l_s$ for all $j\leq n$
	\end{theorem}
	\begin{proof}
		Since there are $l_s$ many possible values, and the algorithm processes a value only once, the algorithm constructs a sum of length at most $l_s$ iff it is not t-inductive.
		It follows from (\autoref{eq:sumprop}) that if (\autoref{Theorem:inductive}) is violated, than it is violated by some vector $x$ with $|x|\leq l_s$.
	\end{proof}
	It follows from \autoref{Theorem:smallx} that deciding inductivity is in co-NP. Choosing some vector $x$ with values at most $l_s$ non-deterministically and checking if it violates \autoref{Theorem:inductive} takes polynomial time.
	\medskip
	
	We assume the dimension $n$ of $k$ (which is the number of places in the Petri net) is a fixed parameter and introduce a new algorithm that solves t-inductivity in logarithmic space. 
	It simply iterates all possible vectors $x$ that satisfy $x(j)\leq l_s$ for all $j\leq n$ and checks whether they satisfy the inequality.
	The successor function $\mathit{succ}$ handles the vector $x\leq l_s$ like a number with $n$ digits to the basis $l_s$ and works like a standard successor. It starts at the first value and if it is less than $l_s$ it adds one and terminates, if the current value is $l_s$, it sets it to $0$ and handles the next one.
	\begin{algorithm}
		\caption{Inductivity-LogSpace}
		\label{logspaceAlg}
			$x=0^{n}$\;
			\Repeat{$x=l_s^{n}$}{
			\If{$c \leq k\mal x + k \mal t^- < c- k\mal t^\Delta$}{
				\Return Not inductive 
			}
			$x=\mathit{succ}_{l_s}( x )$\;
		}
			\Return Inductive
	\end{algorithm}

	The value of $l_s$ is linear in every input variable and thus it can be stored in logarithmic space.
	Any vector $x$ with $x(j)\leq l_s$ for all $j\leq n$ can also be stored in logspace.
	
	\begin{theorem}
		Deciding inductivity of a half space is in \textbf{L} for unary encoded input and fixed dimension of $k$.
	\end{theorem}
	
	A nondeterministic version of the inductivity algorithm has to store only the current value and the number of executed steps and it executes at most $l_s$ steps.
	\begin{theorem}\label{unaryNL}
		Deciding inductivity of a half space is in \textbf{co-NL} for unary encoded input.
	\end{theorem}

	For binary encoded input, it follows from \autoref{unaryNL}:
	\begin{theorem}
		Deciding inductivity of a half space is in \textbf{co-CSL} for binary encoded input.
	\end{theorem}
	
		For an instance of the inductivity problem given by a half space and a transition we introduce the parametrized instance with the greatest total value $k_{max}$ of $k$ as the parameter.
	
	\begin{theorem}\label{Theorem:tractable}
		The parametrized inductivity problem is fixed parameter tractable.
	\end{theorem} 
	\begin{proof}
		For any vector $x\in \N^{n}$ with $x_i\geq k_1$ it holds
		$$
		k\mal x= k \mal (x_1+k_i, \ldots ,x_{i-1},x_i-k_1,x_{i+1}, \ldots ,x_n)^T
		$$
		We iterate this argument and it follows that if there is a vector that satisfies the condition of \autoref{Theorem:inductive} then it is also satisfied by a vector $x$ with $x_2, \ldots x_n\leq k_1$.
		We assume $k\geq 0$ and $k\mal t^- <c$.
		
		$$k\mal t^- + k\mal x=k\mal t^- + k_1\mal x_1+ \sum_{i=2}^n k_i\mal x_i\geq c$$
		$$\Rightarrow k_1\mal x_1 \geq c-k\mal t^- - \sum_{i=2}^n k_i\mal x_i\geq c -k\mal t^- - k_1\mal \sum_{i=2}^n k_i$$
		$$\Rightarrow x_1\geq \lceil \frac{c - k\mal t^- - k_1\mal \sum_{i=2}^n k_i}{k_1} \rceil$$
		Instead of imposing a lower bound on $x_1$ we introduce $x'$ with $x_1=x'_1- \lceil \frac{c -k\mal t^- - k_1\mal \sum_{i=2}^n k_i}{k_1} \rceil$ and $x'_i=x_i$ for $i>1$:
		$$ k\mal x= k \mal x' + k_1 \mal\lceil \frac{c -k\mal t^- - k_1\mal \sum_{i=2}^n k_i}{k_1} \rceil $$
		It follows that an half space $k\mal \amark \geq c$ is t-inductive iff the following half space is t-inductive: $k\mal \amark \geq c-k_1 \mal\lceil \frac{c -k\mal t^- - k_1\mal \sum_{i=2}^n k_i}{k_1} \rceil$.
		Note that the new half space is only bounded by $k_{max}$ and not $c$:
		$$c-k_1 \mal\lceil \frac{c -k\mal t^- - k_1\mal \sum_{i=2}^n k_i}{k_1} \rceil\leq k\mal t^- + k_1\mal \sum_{i=2}^n k_i$$ 
		The construction for $k\leq 0$ and $k\mal t^- >c$ is analogue.
	\end{proof}
	
	\subsection{Proofs for Generating c}\label{sec:pn:app:genc}
	The main contribution of this subsection is the proof of \autoref{corollary_bound_on_c}. This requires the following two technical lemmas.
	\begin{lemma}\label{lem:gcd}
$$|\gcd{k} |\leq | k\mal \trans |$$
\end{lemma}
\begin{proof}
It holds $k\mal \trans =\trans (1) \mal k(1) + \ldots + \trans (n) \mal k(n)=z\mal \gcd{k}$  for some $z\in \mathbb{Z}$. It follows $|k\mal \trans |\geq |\gcd{k} |$.
\end{proof}

	\begin{lemma}\label{theorem_fin_many_c}
		Let $(k,c)$ be a non-trivial $t$-inductive half space and let $y \in \N$ denote the Frobenius number of $\frac{k(1)}{gcd(k)}, \dots, \frac{k(n)}{gcd(k)}$.
		\begin{enumerate}
			\item[a)] \label{theorem_finite_c_part_a} If $k \geq 0$, it holds $k \mal t^{-} + \ktdelta < c \leq gcd(k)\mal y + k \mal t^{-} $.
			\item[b)] \label{theorem_finite_c_part_b} If $k \leq 0$, it holds $gcd(k)\mal y + k\mal  t^{-} \leq c < k\mal t^{-}$.
		\end{enumerate}
	\end{lemma}

\begin{proof} 
	We prove the lower and upper bounds for both cases.
	\begin{itemize}
		\item[a)] The lower bound follows immediately from \autoref{Lemma:Nontrivial}.
		For the upper bound, we assume $c>gcd(k)\mal y + k\mal t^{-}$ and thus $\frac{c-k\mal t^{-}}{gcd(k)}> y$. 
		Since $y$ is the Frobenius number, there is a vector $b\in \N^n$ such that $\lfloor \frac{c-k t^{-}}{gcd(k)}\rfloor=\frac{k^T}{gcd(k)}  \mal b$. 
		This means $$ c-k\mal t^{-}\leq {k}\mal b<  c-k\mal t^{-}+{gcd(k)}.$$	
		According to \autoref{Lemma:Nontrivial}, it holds $k \mal \trans <0$. 
		We apply \autoref{lem:gcd} and get $ -k \mal \trans \geq gcd(k)$. This means that $c\leq k \mal b +k t^{-} <c- \ktdelta$ holds and thus the vector $ b$ satisfies \autoref{Theorem:inductive}. 
		This is a contradiction to inductivity.
		\item[b)] The upper bound follows immediately from \autoref{Lemma:Nontrivial}.
		For the lower bound, we assume $c<gcd(k) \mal y + k \mal t^{-}$ and thus $\frac{c-k t^{-}}{gcd(k)}> y$. 
		The remainder is analogue to a).
	\end{itemize}

%Synce $y$ is the Frobenius number, there is a vector $b\in \N^n$ such that $c-k t^{-}=k^T \mal gcd(k) \mal b$. Since $k \mal \trans <0$, it follows that $c=k^T \mal gcd(k) \mal b +k t^{-} <c- \ktdelta$ holds and thus the vector $gcd(k) \mal b$ satisfies \autoref{eq:thmx}. This is a contradiction to inductivity.
\end{proof}		

We recall \autoref{corollary_bound_on_c}:
	\boundonc*
	\begin{proof}
	Let $y$ denote the Frobenius number of $\frac{k(1)}{\gcd(k)}, \dots, \frac{k(n)}{\gcd(k)}$. By \autoref{Lemma:NecessNontrivial} and \autoref{lem:gcd}, we know that $\abs{k(i)} > -\ktdelta \geq \abs{\gcd(k)}$. From this we obtain that $\frac{k(i)}{\gcd(k)} = \frac{\abs{k(i)}}{\abs{\gcd(k)}} \geq 2$. Now we apply the definition of the Frobenius number and get: $y \leq (\frac{k_{max}}{\gcd(k)} - 1)(\frac{k_{min}}{\gcd(k)} - 1)$. \\
		Now assume that $k \geq 0$. We give an estimation for $\gcd(k) \cdot y$:
		\begin{align*}
			\gcd(k)\cdot y &\leq \gcd(k) \cdot (\frac{k_{max}}{\gcd(k)} - 1)(\frac{k_{min}}{\gcd(k)} - 1)  \\
			&\leq \gcd(k)^{2} \cdot (\frac{k_{max}}{\gcd(k)} - 1)(\frac{k_{min}}{\gcd(k)} - 1) \\
			&\leq (k_{max} - \gcd(k))(k_{min} - \gcd(k)) \\
			&\leq k_{max} \cdot k_{min}.
		\end{align*}
		We now combine this with the bound proven in \autoref{theorem_fin_many_c} and derive the criterion of \autoref{corollary_bound_on_c} $$c < k_{max} \cdot k_{min} + k \mal t^{-}.$$ 
		\medskip
		If $k \leq 0$, we derive a similar bound. Note that it holds $\gcd(k) < 0$). We now derive a lower bound of $	\gcd(k) \cdot y$:
		\begin{align*}
			\gcd(k) \cdot y &\geq \gcd(k) \cdot (\frac{k_{max}}{\gcd(k)} - 1) \mal (\frac{k_{min}}{\gcd(k)} - 1) \\
			&\geq - \gcd(k)^{2} \cdot (\frac{k_{max}}{\gcd(k)} - 1)\mal (\frac{k_{min}}{\gcd(k)} - 1) \\
			&\geq - (k_{max} - \gcd(k))\mal (k_{min} - \gcd(k)) \\
			&\geq - k_{max} \cdot k_{min}.
		\end{align*}
		Like above, we apply this to the bound on $c$ from \autoref{theorem_fin_many_c} and get 
		$$c \geq -k_{max} \cdot k_{min} + k \mal t^{-}$$
		
	\end{proof}
\section{Non-Trivial Petri Nets}
Since the benchmark suite did not require non-trivial separating IHS, it did not accurately present our CEGAR method. 
We would like a better understanding of which Petri nets require non-trivial separating IHS. 
For this purpose, we construct a simple Petri net that has a non-trivial separating IHS but not a trivial one. We begin by collecting sufficient conditions of a Petri net that ensure non-triviality for any separating IHS.
\begin{lemma}
	Let $a\in \Real_+^m$  be such that $m_f=m_0+\sum_{i=1}^m a(i)\mal \tDelta_i$. For any separating half space, there is a transition that is not oriented towards it. 
\end{lemma}
\begin{proof}
	Since the half space is separating, it holds $k\mal m_0\geq c>k\mal m_f$ and thus $0>k\mal (m_f-m_0)=\sum_{i=1}^m a(i)\mal k\mal  \tDelta_i$. 
	One element in the sum has to be negative: $\exists_{i<m} :  a(i)\mal k\mal  \tDelta_i<0$. 
	Since $a(i)$ can not be negative, it follows  $\exists_{i<m} :   k\mal  \tDelta_i<0$.
	So $(k,c)$ is not oriented towards $t_i$.
\end{proof}
% (litmus tests and mutual exclusion algorithms)

It follows that, for any separating half space, one of the transitions $t_i$ with an associated value $a_i$ greater than zero is not oriented towards it.
In order to ensure that the separating half space is not trivial, we require two additional properties: it can neither be antitone, nor monotone.

For any $t_i$ with $a_i>0$ we require $t_i^-\leq m_0$, i.e.\ the transition is activated in the initial marking. This means $\Act{t}\cap \Sol{k}{c}\neq \emptyset$ and thus it is not antitone.

For any $t_i$ with $a_i>0$ we require $t_i^- +\tDelta_i \leq \amark_f$, which means there is a marking from which $t_i$ can be fired in order to reach $\amark_f$. Assume there is separating half space $(k,c)$ that is monotone for $t_i$. Then it holds $k\geq 0$ and thus $k\mal \amark_f\geq k\mal (t_i^-+\tDelta_i)$. Since $k\mal \amark_f <c$, it follows $ k\mal (t_i^-+\tDelta_i)<c$. 
%And thus $m_f\in \Sol{k}y{c}$ and the half space is not separating. 
This is a contradiction to monotonicity for $t_i$. % and thus the half space can not be monotone.

In summary, if the marking equation has a continuous solution such that the used transitions can all be fired from the initial marking and they can all be fired to reach $m_f$, then there are no trivial separating half spaces. 

This sufficient condition for non-triviality is useful, because it is not much stronger than the following necessary condition for unreachability.
 If no solution of the marking equation exists where at least one used transition can be fired in the beginning and one in the end, then $m_f$ is unreachable. 
 This condition is very easy to check.
This comparison suggests that for a Petri net where it is not immediately obvious that $m_f$ is unreachable, a non-trivial half space is likely to be required.

\paragraph*{Example}
We now construct a minimal non-trivial example for larger dimensions. 
We introduce a Petri net $N_n$ of size $n\geq 3$ that has a non-trivial separating half space but no trivial separating half spaces (see~\autoref{Figure:genex}). 
Furthermore, if one transition is removed, a trivial separating invariant exists.

We set $m_0=1^n, m_f=2^n$, meaning we start with one token in each place and ask whether we can avoid getting having tokens in each place. 
Then, we choose some $j\leq n$ and we define $n$ transition $t_1,...t_n$ such that each transition $t_i$ removes one token from each place and then puts $n$ tokens in place $p_i$. The exception is $t_j$ which puts $n+1$ tokens into $p_j$. 

Formally, this means $t_i^-=1^n$ for all $i\leq n$ and $\tDelta_i(i)=n-1, \tDelta_i(k)=-1$ for $k\leq n,k\neq i$. For $t_j$, it holds $\tDelta_j(i)=-1$ for $i\leq n,i\neq j$ and $\tDelta_j(j)=n$.

\begin{figure}
	%\centering
	\begin{center}
	\begin{tikzpicture}[node distance=1.5cm and 0.5cm , bend angle=45, auto, ->]
		\node [place] (p1) {$p_1$};
				\node[right= of p1] (dot1) {$\cdots$};
		\node[place, right= of  dot1] (pi) {$p_i$};
				\node[right= of  pi] (doti) {$\cdots$};
		\node[place, right= of  doti] (pj) {$p_j$};
				\node[right= of  pj] (dotj) {$\cdots$};
		\node [transition, above= of  pi] (ti) {$t_i$};	
		\node [transition, below= of  pj] (tj) {$t_j$};
		
	%	\coordinate[right= of  p1] (tco);
			\path (p1) edge[out = 90, in = 180] (ti);
			\path (p1) edge[out = 285, in = 180] (tj);

			\draw[->] (pi) edge (ti);
			\draw[->, bend right] (ti) edge node[left] {\footnotesize{$n$}} (pi);
			\draw[->] (pi) edge (tj);
			\draw[->] (pj) edge (ti);
			\draw[->] (pj) edge (tj);
			\draw[->, bend right] (tj) edge node[right] {\footnotesize{$n+1$}} (pj);		
		\node[place, right= of  dotj] (pn) {$p_n$};
			\path (pn) edge[out = 270, in = 0] (tj);
			\path (pn) edge[out = 90, in = 0] (ti);

	\end{tikzpicture}
\end{center}
	\caption{A non-trivial Petri net}\label{Figure:genex}
\end{figure}
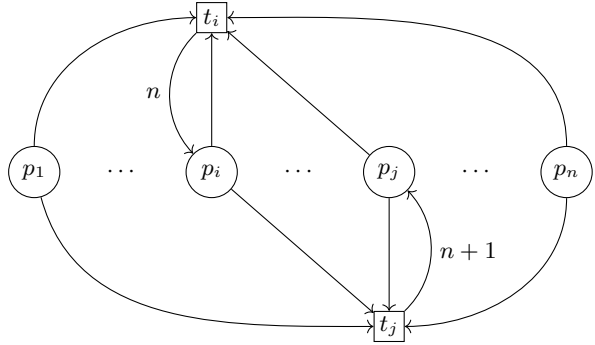

\begin{lemma}
	The Petri net $N_n$ has a separating non-trivial IHS but no separating trivial IHS.
\end{lemma}

\begin{proof}
Let $c=-n\mal (n+1)$, $k(j)=-n$, and $k(i)=-(n+1)$ for $i\neq j,i\leq n$. 
The half space is separating, since it holds $k\mal m_0=-(n+1)\mal (n-1)-n= -n\mal (n+1)+1 > c$ and $k\mal m_f=-(n+1)\mal (n-1)\mal 2-2n=-2(n+1)\mal n +2< c$.

We show that the half space is a IHS using~\autoref{Theorem:thm1}.
Since $k\mal \tDelta_j=(n+1)\mal (n-1)-n^2=-1$ it holds $k\mal \tMinus_j=k\mal \amark_0 \geq c-k \mal \tDelta_j$.
If we add any $k_i$ then we get  $k\mal \tMinus_j +k_i = -n\mal (n+1)+1-n<c$ and if we add $k_j$ or additional values of $k$ we get an even smaller value.

Let $(k,c)$ be any separating half space. Since $k\mal m_0>k\mal m_f$, it holds $-k_1\ldots -k_n>0$. 
Let $k_l=min(k_1,\ldots k_n)$ be the negative entry of $k$ with the largest absolute value. If $l=j$ , then $k\mal \tDelta_l=-k_1...-k_n+n\mal k_l+k_l\leq -n\mal k_l+n \mal k_l+k_l=k_l<0$. 

If $l\neq j$, then $k\mal t_l=-k_1...-k_n+n\mal k_l=k_l$. If all entries of $k$ are not equal, then it holds $-k_1+...-k_n > n\mal k_l$ and thus $k\mal t_i<0$. 
If all entries of $k$ are equal then they are also negative and it holds 
$$k\mal t_j=-k_1...-k_n+n\mal k_j+k_j=-n\mal k_j+n\mal k_j+k_j=k_j<0.$$ 
It follows that any separating half space is oriented towards at least one transition.

Obviously all transitions are enabled at $m_0$ and the half space is not antitone.
According to $-k_1\ldots -k_n>0$, it holds $k\ngeq 0$ and thus it is not monotone either.
\end{proof}

%We divide our results into two categories: reachability and coverability checks.

We evaluate the performance of \inequalizer\ for reachability on the non-trivial Petri nets of sizes three to ten in~\autoref*{tab:nontrivial}. We give the number of iterations of the CEGIS loop performed by the tool. 
We do not include the run-time results of \mist\ for these Petri nets in the table since they were all well below $0.1$ second. We use incremental solving and find that our tool usually computes the IHS quickly using few iterations. There are only two diverging results where the SMT-solver returns a number of unusable vectors $k$.

\begin{table}[t]
	\centering

% !TEX spellcheck = en-EN

\begin{tabular}{|r|r|r|}
\hline
$|P|$ & Iterations & Time \\
\hline
%\bench{2}& 5& 4 & 0.6 & 0.1 & 0.1 & 0.1 & 0.1 & 0.1 \\
\bench{3}  & 2 & 0.4  \\
\bench{4} & 113 & 6.9 \\
\bench{5}  & 2 & 0.4 \\
\bench{6} & 2 & 0.4 \\
\bench{7}  & 6 & 0.6 \\
\bench{8}  &3 & 0.6  \\
\bench{9}  &378 & 205.2 \\ 
\bench{10} & 2  & 0.5 \\
\hline
\end{tabular}
	
	\caption{\inequalizer\ on non-trivial Petri nets. }
	\label{tab:nontrivial}
\end{table}

\end{document}